\documentclass[onecolumn,final]{IEEEtran}

\usepackage{tikz}
\usetikzlibrary{}
\usepgfmodule{oo}

\usepackage[]{hyperref}
\usepackage{amsthm}
\usepackage{amsmath}
\usepackage{amssymb}
\usepackage{quantum}
\usepackage{stmaryrd}
\usepackage{graphicx}
\allowdisplaybreaks[4]

\begin{document}

\pgfdeclarelayer{background}
\pgfdeclarelayer{firstbackground}
\pgfdeclarelayer{secondbackground}
\pgfsetlayers{secondbackground,firstbackground,background,main}

\pgfooclass{stamp}{ % This is the class stamp
    \method stamp() { % The constructor 
    }
    \method apply(#1,#2,#3) { % Causes the stamp to be shown at coordinate (#1,#2)
        %Draw the stamp:
	\draw (#1+2,#2) -- (#1,#2) -- (#1,#2+1) -- (#1+2,#2+1) -- cycle;
          \node[font=\tiny] at (#1+1,#2+0.5) {#3};
    }
    \method box(#1,#2,#3,#4,#5) { % Causes the stamp to be shown at coordinate (#1,#2)
        %Draw the stamp:
	\filldraw[fill=#5] (#1+#3,#2) -- (#1,#2) -- (#1,#2+#4) -- (#1+#3,#2+#4) -- cycle;
   }
 \method cnot(#1,#2,#3,#4) { % Causes the stamp to be shown at coordinate (#1,#2)
        %Draw the stamp:
	\draw (#1,#2) -- (#1-#3,#2) -- (#1-#3,#2-#4) -- (#1,#2-#4)-- (#1-#3-2,#2-#4);
	\draw (#1-#3,#2) circle (0.2);
	\draw (#1-#3,#2+0.2) -- (#1-#3,#2) -- (#1-#3-0.25,#2);
	\draw (#1-#3-0.25,#2) -- (#1-#3-2,#2);
   }
}
\pgfoonew \mystamp=new stamp()

\pgfooclass{block}{ % This is the class stamp
    \method block() { % The constructor 
    }
    \method basic(#1,#2,#3) { % Causes the stamp to be shown at coordinate (#1,#2)
	\mystamp.apply(#1,#2,#3)
	\draw (#1-1,#2+0.5) -- (#1,#2+0.5);
	\draw (#1+2,#2+0.5) -- (#1+3,#2+0.5);
    }
}

\pgfoonew \myblock=new block()

\title{Bounds on Information Combining\\With Quantum Side Information}

\author{Christoph~Hirche and~David~Reeb% <-this % stops a space
\thanks{C. Hirche is with the F\'{\i}sica Te\`{o}rica: Informaci\'{o} i Fen\`{o}mens
 Qu\`{a}ntics, 
 Departament de F\'{i}sica, 
Universitat Aut\`{o}noma de Barcelona,   ES-08193
 Bellaterra (Barcelona), Spain.}% <-this % stops a space
\thanks{D. Reeb is with the Insitute for Theoretical Physics, Leibniz Universitat Hannover, Appelstrasse 2, 30167 Hannover, Germany.}% <-this % stops a space
}

\maketitle

\begin{abstract}
``Bounds on information combining'' are entropic inequalities that determine how the information (entropy) of a set of random variables can change when these are combined in certain prescribed ways. Such bounds play an important role in classical information theory, particularly in coding and Shannon theory; entropy power inequalities are special instances of them. The arguably most elementary kind of information combining is the addition of two binary random variables (a CNOT gate), and the resulting quantities play an important role in Belief propagation and Polar coding. We investigate this problem in the setting where quantum side information is available, which has been recognized as a hard setting for entropy power inequalities.

Our main technical result is a non-trivial, and close to optimal, lower bound on the combined entropy, which can be seen as an almost optimal ``quantum Mrs.\ Gerber's Lemma''. Our proof uses three main ingredients: (1) a new bound on the concavity of von Neumann entropy, which is tight in the regime of low pairwise state fidelities; (2) the quantitative improvement of strong subadditivity due to Fawzi-Renner, in which we manage to handle the minimization over recovery maps; (3) recent duality results on classical-quantum-channels due to Renes et al. We furthermore present conjectures on the optimal lower and upper bounds under quantum side information, supported by interesting analytical observations and strong numerical evidence.

We finally apply our bounds to Polar coding for binary-input classical-quantum channels, and show the following three results: (A) Even non-stationary channels polarize under the polar transform. (B) The blocklength required to approach the symmetric capacity scales at most sub-exponentially in the gap to capacity. (C) Under the aforementioned lower bound conjecture, a blocklength polynomial in the gap suffices.
\end{abstract}
%\begin{IEEEkeywords}
%
%\end{IEEEkeywords}

\section{Introduction}
Many of the tasks in classical and quantum information theory are concerned with the evolution of random variables and their corresponding entropies under certain ``combining operations''. 
A particularly elementary example is the addition of two classical random variables (with values in some group). In this case the entropy can be easily computed since we know that the addition of two random variables has a probability distribution which corresponds to the convolution of the probability distributions of the individual random variables. 
The picture changes when we have random variables \emph{with} side information. Now we are interested in the entropy of the sum conditioned on all the available side information. Evaluating this is substantially more difficult, already in the case of classical side information. 

The field of \textit{bounds on information combining} is concerned with finding optimal entropic bounds on the conditional entropy in this and other ``information combining'' scenarios. 
A first lower bound was given by Wyner and Ziv in \cite{WZ73}, the well known \textit{Mrs. Gerbers Lemma}, which immediately found many applications (see e.g. \cite{GKbook}). 

Following these results, additional approaches to the problem have been found which also led to an \emph{upper} bound on the conditional entropy of the combined random variables. One proof method and several additional applications can be found e.g. in \cite{RU08} along with the optimal upper bound. 

Now, we are interested in above setting, but with \emph{quantum} -- rather than classical -- side information. Unfortunately, it turns out that none of the classical proof techniques apply in this quantum setting, since conditioning on quantum side information does not generally correspond to a convex combination over unconditional situations. In this work we are concerned with investigating the optimal entropic bounds under quantum side information and report partial progress along with some conjectures. 

An alternative way of looking at the problem is by associating the random variables along with the side information to channels, where the random variable models the input of the channel leading to a known output given by the side information. This analogy is especially useful when investigating coding problems for classical channels. Recently, Arikan~\cite{A09} introduced the first example of constructive capacity achieving codes with efficient encoding and decoding, Polar codes. The elementary idea of Polar codes is to combine two channels by a CNOT gate at their (classical) input, which means that the input of the second channels gets added to the input of the first channel, therefore adding additional noise on the first input, but providing assistance when decoding the second channel. 
To evaluate the performance of these codes, the Mrs.\ Gerbers Lemma provides an essential tool to tracking the evolution of the entropy through the coding steps (see e.g.\ \cite{AT14,GX15}, which we will build on below).
Following their introduction in the classical setting, Polar codes have been generalized to classical-quantum channels~\cite{WG13}. Finding good bounds on information combining with quantum side information can therefore be very useful for proving important properties of classical-quantum Polar codes, as we will see.

In this work we provide a lower bound on the conditional entropy of added random variables with quantum side information, using novel bounds on the concavity of von Neumann entropy, recent improvements of strong subadditivity by Fawzi and Renner~\cite{FR14}, and results on channel duality by Renes \etal~\cite{RSH14, R17}. Furthermore we provide a conjecture on the optimal inequalities (upper and lower bounds) in the quantum case. Then we discuss applications of our technical results to other problems in information theory and coding; in particular, we show how to use our results to prove sub-exponential convergence of classical-quantum Polar codes to capacity, and that polarization takes place even for non-stationary classical-quantum channels. 

\subsection{Entropy power inequalities}\label{EPI}
Bounds on information combining are generalizations of a family of entropic inequalities that are called \emph{entropy power inequalities} (for historic reasons). The first and paradigmatic of these inequalities was suggested by Shannon in the second part of his original paper on information theory~\cite{S48}, stating that
\begin{align}\label{shannon-epi}
 e^{2h(X_1)/n}+e^{2h(X_2)/n}\leq e^{2h(X_1+X_2)/n},
\end{align}
where $X_1,X_2$ and are random variables with values in ${\mathbb R}^n$ and $h$ denotes the differential entropy (each of the three terms in (\ref{shannon-epi}) is called the \emph{entropy power} of the respective random variable $X_1$, $X_2$, and $X_1+X_2$); rigorous proofs followed later~\cite{S59}. Clearly, this inequality gives a lower bound on the entropy $h(X_1+X_2)$ of the sum $X_1+X_2$ given the individual entropies $h(X_1),h(X_2)$, and it is easy to see that the bound is tight (namely, for Gaussian $X_1,X_2$).

Similar lower bounds on the entropy of a sum of two (or more) random variables with values in a group $(G,+)$ have also been termed entropy power inequalities, see e.g.\ \cite{SW90}. For the simplest group $G={\mathbb Z}_2$, the optimal lower bound follows from a famous theorem in information theory, called \emph{Mrs.\ Gerber's Lemma} \cite{WZ73}, which we will describe below in more detail. For the group $G={\mathbb Z}$ of integers, entropy power inequalities in the form of lower bounds on the entropy have emerged~\cite{HAT14, Tao10} after a combinatorial version of the question had been investigated in the field of arithmetic (and in particular, additive) combinatorics for a long time.

Most (or all) of the above entropy power inequality-like lower bounds remain valid when classical side information $Y_i$ is available for each of the random variables $X_i$, so that for example the entropic terms in (\ref{shannon-epi}) are replaced by $h(X_1|Y_1)$, $h(X_2|Y_2)$, and $h(X_1+X_2|Y_1Y_2)$, respectively. This is due to typical convexity properties of these lower bounds along with a representation of the conditional Shannon entropy as a convex combination of unconditional entropies (see our description of the classical conditional Mrs.\ Gerber's Lemma below).

Entropy power inequalities have recently been investigated in the quantum setting \cite{KS14, PMG14,ADO15}, with the action of addition replaced by some quantum combining operation, such as a beamsplitter operation on continuous-variable states, or a partial swap. These inequalities also hold under conditioning on \emph{classical} side information.
 
However, when the side information is of \emph{quantum} nature, i.e.\ each $(X_i,Y_i)$ is a classical-quantum state \cite{NC00} (for the classical entropy power inequalities) or a fully quantum state (for the quantum entropy power inequalities), the proofs do not anymore go though in the same way. Actually, as we will see in this paper, the inequalities that hold under classical side information can sometimes be violated in the presence of quantum side information.

The only lower bounds available under quantum side information so far can be found in~\cite{K15, dPT17}, where for (Gaussian) quantum states an entropic lower bound was proven for the beamsplitter interaction. No general results for all classical-quantum states have been obtained so far.

In the light of these developments, our contribution can be seen as the natural entry point into investigating the influence of \emph{quantum} side information in entropy power inequalities and information combining: For the ``information part'' we concentrate on the simplest scenario, namely \emph{classical} random variables $X_i$ that are \emph{binary-valued}, i.e.\ valued in the simplest non-trivial group $({\mathbb Z}_2,+)$. For the side information $Y_i$, however, we allow any general quantum system and states. Our question therefore highlights the added difficulties coming from the quantum nature of side information.

But already this bare scenario gives new results on highly relevant coding scenarios: The entropic lower bounds we prove are enough to guarantee the polarization of classical-quantum Polar codes, even with a guaranteed speed of polarization in the i.i.d.\ case. Furthermore, we conjecture optimal upper and lower bounds on information combining in this simple scenario which exhibit interesting properties, as we will describe.

\section{Preliminaries}\label{Pre}
In this section we will introduce the necessary notation along with some definitions and simple observations. 
In the following the systems in question will be modeled by random variables, where the main classical random variable is usually denoted by $X_i$ (equivalent to the input system of a channel). Whenever the random variables modeling the side information (or the channel output) are classical we will denote them by $Y_i$, and when they are quantum then usually by $B_i$ (a good part of the formalism applies to both situations).

The main quantity under investigation will be the von-Neumann entropy, for a quantum state $\rho_A$ on a quantum system $A$ defined by
\begin{equation}
H(A) = -\tr \rho \log\rho,
\end{equation}
which reduces to the Shannon entropy in the case of classical states (those which are diagonal in the computational basis). In this paper, we leave the base of the logarithm unspecified, unless stated otherwise, so that the resulting statements are valid in any base (like binary, or natural); our figures, however, use the natural logarithm. A particular case is the Shannon entropy for a binary probability distribution, called the binary entropy $h_2(p)=-p\log{p} - (1-p)\log{(1-p)}$. In the following will often use the inverse of this function
\begin{equation}
h_2^{-1}:[0,\log2]\to[0,1/2]. 
\end{equation}
From here the conditional entropy of a quantum state $\rho_{AB}$ on a quantum system $AB$ is defined by 
\begin{equation}
H(A|B) = H(AB) - H(B). 
\end{equation}
Whenever the conditioning system is classical we can state the following important property
\begin{equation}
H(A|Y) = \sum_{y} p(y) H(A|Y=y), 
\end{equation}
this obviously holds also for the Shannon entropy, but importantly we cannot write down such a decomposition when the conditioning system is quantum.

For a binary random variable $X$ we can associate a probability distribution $p$ for which then $H(X) = h_2(p)$. Now it is well known that when we sum two random variables the corresponding probability distribution is the convolution of the original probability distributions. The binary convolution is defined as $a\ast b := a(1-b) + (1-a)b$. It easily follows that
\begin{equation}
H(X_1+X_2) = h_2(h_2^{-1}(H(X_1))\ast h_2^{-1}(H(X_2))).
\end{equation}

Often we want to stress the duality of classical-quantum states to classical-quantum channels. In this case, for a given channel $W$ with input modeled by a random variable $X$ and the output by $Y$ we write equivalently 
\begin{equation}
H(X|Y) = H(W)
\end{equation}
(usually we assume here the uniform distribution over input values $X$). An additional useful entropic quantity is the mutual information defined as 
\begin{equation}
I(X:Y) = H(X) + H(Y) - H(XY).
\end{equation} 
Again for a channel $W$ with binary classical input, we write $I(W)$, which, when fixing $X$ to correspond to the uniform probability distribution, is also called the \emph{symmetric capacity} of that channel:
\begin{equation}
I(W) = I(X:Y) = \log 2 - H(X|Y)  = \log 2 - H(W).
\end{equation}
Lastly, mostly for technical purposes we will also use the relative entropy defined as 
\begin{equation}
D(\rho\|\sigma) = \tr \rho \left(\log\rho -\log\sigma\right), 
\end{equation} 
which in the case of classical probability distributions is the Kullback-Leibler divergence.

\section{Bounds on Information Combining in classical information theory}

In classical information theory the topic of bounds on information combining describes a number of results concerned with what happens -- in particular, to the entropy of the involved objects -- when random variables get combined. This is especially interesting when we have side information for these random variables, due to the analogy with channel problems.
The name of this field goes back to~\cite{LHHH05} where such bounds where used for repetition codes. Later on, many more results where found, also as \textit{Extremes of information combining}~\cite{SSZ05} for MAP decoding and LDPC codes. Furthermore, information combining plays an important role for belief propagation~\cite{RU08}. 

Examples of particular importance are the combinations at the variable and check nodes in belief propagation and the transformation to \textit{better} and \textit{worse} channels in Polar coding. For the latter, two (classical) binary random variables with side information are added, or equivalently two binary channels get combined by a CNOT gate (see Figure~\ref{Fig:Basic}). 
In the first setting we are concerned with the entropy of the sum $X_1+X_2$ given the side information $Y_1Y_2$, which corresponds to check nodes in belief propagation and the worse channel in Polar coding. In the channel picture this can be seen as channel combination
\begin{equation}\label{cl:minus}
(W_1\boxast W_2)(y_1y_2|u_1) = \frac{1}{2}\sum_{u_2} W_1(y_1 | u_1 \oplus u_2)W_2(y_2|u_2),
\end{equation} 
and is therefore given by
\begin{equation}
H(X_1+X_2|Y_1Y_2) = H(W_1\boxast W_2).
\end{equation}
In the second setting we are interested in the entropy evolution at a variable node, with output states given by
\begin{equation}\label{cl:plus}
(W_1\varoast W_2) (y_1y_2|u_2)  = W_1(y_1 | u_2)W_2(y_2|u_2).
\end{equation} 
It turns out that the combined channel can be reversibly transformed (see e.g.~\cite{R16bp}) into a channel with the output states
\begin{equation}\label{cl:plus2}
u_2\rightarrow \frac{1}{2} W(y_1 | u_1 \oplus u_2)W(y_2|u_2),
\end{equation}  
which is equivalent to decoding the second input to two channels combined by a CNOT gate given the side information $Y_1Y_2$ but additionally $X_1+X_2$. This again is equal to the generation of a \textit{better} channel studying Polar codes.

Therefore we are interested in the entropy
\begin{equation}\label{eq:varo-minus}
H(X_2 | X_1+X_2,Y_1Y_2) = H(W_1\varoast W_2).
\end{equation}

Lower and upper bounds on both of these quantities have many applications in classical information theory, e.g. in coding theory giving exact bounds on EXIT charts~\cite{RU08} and, of course, the investigation of Polar codes~\cite{AT14,GX15}. \\
In classical information theory the optimal bounds are well known as follows:
\begin{alignat}{2}
h_2(h_2^{-1}(H_1)\ast h_2^{-1}(H_2)) &\leq H(X_1+X_2|Y_1Y_2) &&\leq \log 2 - \frac{(\log 2 - H_1)(\log 2 - H_2)}{\log 2}, \label{minus-bounds}\\
\frac{H_1 H_2}{\log 2} &\leq H(X_2 | X_1+X_2,Y_1Y_2) &&\leq H_1 + H_2 - h_2(h_2^{-1}(H_1)\ast h_2^{-1}(H_2)),  \label{plus-bounds}
\end{alignat}
with $H_1 = H(X_1|Y_1)$ and $H_2 = H(X_2|Y_2)$. \\
Later in this work we will be particularly interested in the lower bound in \ref{minus-bounds} (and equivalently the upper bounds in \ref{plus-bounds}), which are also known under the name Mrs. Gerbers Lemma. We will review the proofs  of these inequalities in the next subsection, also to show difficulties when translating these inequalities to the quantum setting.
A well known fact is that 
\begin{equation}\label{additiveentropies}
H(X_1+X_2|Y_1Y_2) + H(X_2 | X_1+X_2,Y_1Y_2) = H(X_1|Y_1) + H(X_2|Y_2). 
\end{equation}
From this follows that it is sufficient to prove the inequalities for either Equation~\ref{minus-bounds} or ~\ref{plus-bounds}. We will therefore mostly focus on the setting leading to Equation~\ref{minus-bounds}. \\
Moreover, it is even known for which channels equality is achieved in above equations (see e.g. \cite{RU08}). For the lower bound in Equation \ref{minus-bounds} this is the binary symmetric channel (BSC) and for the upper bound it is the binary erasure channel (BEC).
Therefore these channels are sometimes called the most and least informative channels.

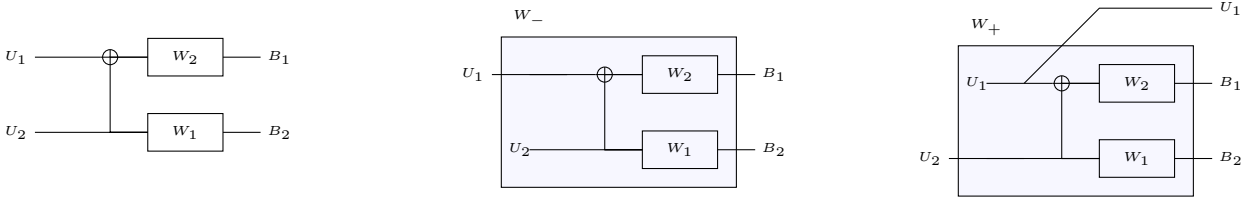
\begin{figure}
\begin{minipage}[ht]{0.33\textwidth}
\begin{tikzpicture}[scale=0.5]
	\myblock.basic(0, 0, $W_1$)
	\myblock.basic(0, 2, $W_2$)
	\mystamp.cnot(0,2.5,1,2)
%    \begin{pgfonlayer}{secondbackground}
%	\mystamp.box(-1.5,-0.5,4,4,blue!3);
%    \end{pgfonlayer}
%        \node[font=\tiny] at (0.5,4) {$W_2$};
	\node[font=\tiny] at (-3.5,0.5) {$U_2$};
	\node[font=\tiny] at (-3.5,2.5) {$U_1$};
        \node[font=\tiny] at (3.5,2.5) {$B_1$};
        \node[font=\tiny] at (3.5,0.5) {$B_2$};
\end{tikzpicture}
\end{minipage}
\begin{minipage}[ht]{0.33\textwidth}
\begin{tikzpicture}[scale=0.5]
	\myblock.basic(0, 0, $W_1$)
	\myblock.basic(0, 2, $W_2$)
	\mystamp.cnot(0,2.5,1,2)
    \begin{pgfonlayer}{secondbackground}
	\mystamp.box(-3.75,-0.5,6.25,4,blue!3);
    \end{pgfonlayer}
        \node[font=\tiny] at (-3,4) {$W_-$};
	\node[font=\tiny] at (-3.25,0.5) {$U_2$};
	\node[font=\tiny] at (-4.5,2.5) {$U_1$};
	\draw (-4,2.5) -- (-2,2.5);
        \node[font=\tiny] at (3.5,2.5) {$B_1$};
        \node[font=\tiny] at (3.5,0.5) {$B_2$};
\end{tikzpicture}
\end{minipage}
\begin{minipage}[ht]{0.33\textwidth}
\begin{tikzpicture}[scale=0.5]
	\myblock.basic(0, 0, $W_1$)
	\myblock.basic(0, 2, $W_2$)
	\mystamp.cnot(0,2.5,1,2)
    \begin{pgfonlayer}{secondbackground}
	\mystamp.box(-3.75,-0.5,6.25,4,blue!3);
    \end{pgfonlayer}
        \node[font=\tiny] at (-3,4) {$W_+$};
	\node[font=\tiny] at (-4.5,0.5) {$U_2$};
	\node[font=\tiny] at (-3.25,2.5) {$U_1$};
	\draw (-4,0.5) -- (-2,0.5);
	\draw (-2,2.5) -- (-0,4.5) -- (3,4.5);
        \node[font=\tiny] at (3.5,2.5) {$B_1$};
        \node[font=\tiny] at (3.5,0.5) {$B_2$};
        \node[font=\tiny] at (3.5,4.5) {$U_1$};
\end{tikzpicture}
\end{minipage}
\caption{\label{Fig:Basic}A useful figure to understand the concept of information combining is to look at two channels $W_1$ and $W_2$ which get combined by a CNOT gate, as in the figure on the left. From this we can generate to types of channels, which, in analogy to Polar coding, we call $W^-$ and $W^+$. Both are directly related since the overall entropy is conserved under combining channels in this way (see Eq.\ (\ref{additiveentropies})).  }
\end{figure}

\subsection{Proof techniques for the classical bounds}
In this section we will review the classical Mrs.\ Gerbers Lemma~\cite{WZ73} and a corresponding upper bound for combining of classical information, in order to contrast these results and proofs with our later results, where conditioning on \emph{quantum} side information is allowed. The following proof sketches illustrate that the classical proofs, which crucially use that the conditional Shannon entropy is affine under conditioning, cannot be easily extended to the case of quantum side information.

\begin{lem}[Mrs. Gerbers Lemma]
Let $(X_1, Y_1)$ and $(X_2,Y_2)$ be independent pairs of classical random variables. Then:
\begin{equation}
H(X_1 + X_2 | Y_1Y_2) \geq h_2(h_2^{-1}(H(X_1 | Y_1))\ast h_2^{-1}(H(X_2 | Y_2))).
\end{equation}
\end{lem}

An important ingredient in the proof of Wyner and Ziv~\cite{WZ73} is the observation that the function
\begin{equation}
g_c (H_1,H_2) := h_2(h_2^{-1}(H_1)\ast h_2^{-1}(H_2))
\end{equation}
is convex in $H_1\in[0,\log2]$ for each fixed $H_2\in[0,\log2]$, and, by symmetry, convex in $H_2$ for each fixed $H_1$. These convexity properties, together with the representation of the conditional Shannon entropy as an average over Shannon entropies, give a proof of the lemma as follows:
\begin{align}
H(X_1 + X_2 | Y_1Y_2) &= \sum_{y_1,y_2} p(Y_1=y_1) p(Y_2=y_2) H(X_1 + X_2 | Y_1=y_1 Y_2=y_2) \label{ccond}\\
&= \sum_{y_1,y_2} p(Y_1=y_1) p(Y_2=y_2) h(h^{-1}(H(X_1 | Y_1=y_1))\ast h^{-1}(H(X_2 | Y_2=y_2))) \\
&\geq \sum_{y_1} p(Y_1=y_1) h(h^{-1}(H(X_1 | Y_1=y_1))\ast h^{-1}(H(X_2 | Y_2))) \\
&\geq h(h^{-1}(H(X_1 | Y_1))\ast h^{-1}(H(X_2 | Y_2))).
\end{align}
Note that the way in which conditioning is handled by the equality \ref{ccond} plays a crucial role in the proof. Unfortunately, this equality does generally not hold for the conditional entropy with quantum side information, i.e.\ when $Y_1$, $Y_2$ are quantum systems; in this case it is not even clear what the correct generalization of the right-hand-side of \ref{ccond} may be. Understanding conditioning on quantum systems is an important but apparently difficult question in quantum information theory, as is drastically illustrated by the much higher difficulty in proving the strong subadditivity property for quantum entropy~\cite{Lieb2002} compared to Shannon entropy. Better understanding of conditioning on quantum side information would not only help for bounds on information combining but for many other open problems as well, like the related question of conditional entropy power inequalities (see Section~\ref{EPI}) or even quantum cryptography~\cite{DFR16}.

In the proof for the upper bound in Equation~\ref{minus-bounds} we encounter a very similar problem handling quantum conditional information. The important inequality for the upper bound is the fact that the function $g_c (H_1,H_2)$ defined above can be bounded by an expression that is affine in both $H_1$ and $H_2$ separately:
\begin{equation}
g_c (H_1,H_2) \leq \log 2 - \frac{(\log 2 - H_1)(\log 2 - H_2)}{\log 2}.
\end{equation}
This follows immediately from the convexity of $g_c$ in $H_1$ and the fact that the inequality holds with equality for each fixed $H_2$ at the two endpoints $H_1\in\{0,\log2\}$, see e.g.\ \cite{LHHH05}. From here, the proof of the classical inequality proceeds in a similar fashion as for the lower bound, using again the expression of the conditional Shannon entropy:
\begin{align}
H(X_1 + X_2 | Y_1Y_2) &= \sum_{y_1,y_2} p(Y_1=y_1) p(Y_2=y_2) h(h^{-1}(H(X_1 | Y_1=y_1))\ast h^{-1}(H(X_2 | Y_2=y_2))) \\
&\leq \sum_{y_1,y_2} p(Y_1=y_1) p(Y_2=y_2) [ \log 2 - \frac{(\log 2 - H(X_1 | Y_1=y_1))(\log 2 - H(X_2 | Y_2=y_2))}{\log 2} ] \\
& =  \log 2 - \frac{(\log 2 - H(X_1 | Y_1))(\log 2 - H(X_2 | Y_2))}{\log 2}.
\end{align}

\section{Information combining with quantum side information}\label{qcombining}

%\textcolor{blue}{Maybe here an introduction to the problem with quantum side information? That way we can also define all necessary quantities and then introduce duality before using it in the next section?} \\
%{\color{red}For reference: Hey Christoph, we can use the Joe-Renes-like argument (which you wrote up before) to ``symmetrize'' our bound as follows. The argument shows directly that the set of points $(H_1,H_2,H(X_1+X_2|B_1B_2))\in{\mathbb R}^3$ has certain symmetries in addition to exchanging $H_1$ and $H_2$, and thus every ``good'' lower bound can be improved to possess the same symmetries. For the following, we should probably reference the Renes paper or your arguments in the right sections. Do we have to re-order our sections (since we are using stuff from future sections)? (Probably everything in the following could be done by looking at purifications of $\rho^{X_1B_1}$ and $\rho^{X_2B_2}$ and without using channels and $\varoast/\boxast$, but this may be just a repetition of some of yours or Joe Renes' stuff.)}

In this section we introduce the generalized scenario of information combining with quantum side information. 
The main ingredient are generalizations of the channel combinations in Equations~\ref{cl:minus} and~\ref{cl:plus} to the case of quantum outputs. Now we are combining two classical-quantum channels, with uniformly distributed binary inputs $\{0,1\}$. Again we will look at both, variable and check nodes under belief propagation and better and worse channels in Polar coding.  Since the inputs are classical we can investigate the same combination procedure via CNOT gates. Belief propagation for quantum channels has been recently introduced in~\cite{R16bp}, for Polar coding the resulting channels can be seen as special case of those in~\cite{WG13}.
 
The generalization of Equation~\ref{cl:minus}, where we look at a check note or equivalently try to decode the input of the first channel while not knowing that of the second becomes a channel with output states
\begin{equation}\label{qu:minus}
W_1\boxast W_2 : u_1 \rightarrow \frac{1}{2}\sum_{u_2} \rho_{u_1 \oplus u_2}^{B_1}\otimes\rho_{u_2}^{B_2}.
\end{equation} 
Similarly the generalization of Equation~\ref{cl:plus} for a variable node is given by
\begin{equation}\label{qu:plus}
W_1\varoast W_2: u_2 \rightarrow \rho_{u_2}^{B_1}\otimes\rho_{u_2}^{B_2},
\end{equation} 
which, by a similar argument then in the classical case is equivalent up to unitaries to the Polar coding setting where we try to decode the second bit while assuming the first bit to be known. This becomes a channel with output states
\begin{equation}\label{qu:plus2}
u_2 \rightarrow  \frac{1}{2}\sum_{u_1}\ketbra{U_1}{u_1}{u_1}\otimes \rho_{u_1 \oplus u_2}^{B_1}\otimes\rho_{u_2}^{B_2},
\end{equation} 
where the additional classical register $U_1$ is used to make the input of the first channel available to the decoder. 

Our goal now is to find bounds on the conditional entropy of those combined channels
\begin{equation}\label{qu:boxast}
H(X_1+X_2|B_1B_2) = H(W_1\boxast W_2),
\end{equation}
and
\begin{equation}\label{qu:varoast}
H(X_2 | X_1+X_2,B_1B_2) = H(W_1\varoast W_2),
\end{equation} 
 in terms of the entropies of the original channels, analog to the bounds on information combining in the classical case. 
An important relation between these two entropies can be directly translated to the setting with quantum side information~\cite{WG13}
\begin{equation}\label{qchain}
H(X_1+X_2|B_1B_2) + H(X_2 | X_1+X_2,B_1B_2) = H(X_1|B_1) + H(X_2|B_2). 
\end{equation}
From here it follows that, as in the classical case, proofing bounds on the entropy in Equation~\ref{qu:boxast} automatically also gives bounds on the one in Equation~\ref{qu:varoast}. 

In the remainder of this section we will introduce the concept of channel duality and discuss its application to channel combining which will help us later to find better bounds on above quantities.

\subsection{Duality of classical and classical-quantum channels}\label{duality}
The essential idea is to embed a classical channel into a quantum state, take its Stinespring dilation and trace over the original output system. In the way we use it here it has been first used in \cite{WR12} to extend classical polar codes to quantum channels and then has been refined in \cite{RSH14} to investigate properties of polar codes for classical channels. A comprehensive overview with some new applications has recently been given in~\cite{R17}.
We explain the procedure here by applying it to a general binary classical channel $W$ with transition probabilities $W(y|x)$. 
The first step is to embed the channels into a quantum state 
\begin{equation}
\varphi_x = \sum_y W(y|x) \ketbra{}{y}{y}
\end{equation}
and then choose a purification of this state with
\begin{equation}
\ket{}{\varphi_x} = \sum_y \sqrt{W(y|x)}\ket{}{y}\ket{}{y}.
\end{equation}
Now we can define our classical quantum channel by an isometry acting as follows
\begin{equation}
U \ket{}{x} = \ket{}{\varphi_x}\ket{}{x}.\label{Udual}
\end{equation}
The dual channel is now defined by the isometry acting on states of the form $\ket{}{\tilde x} = \frac{1}{\sqrt{2}}\sum_z (-1)^{xz} \ket{}{z}$, 
\begin{align}
U \ket{}{\tilde x} &= \frac{1}{\sqrt{2}}\sum_z (-1)^{xz} \ket{}{\varphi_z} \ket{}{z} \\
&= \frac{1}{\sqrt{2}}\sum_{y,z} (-1)^{xz} \sqrt{W(y|x)} \ket{}{y}\ket{}{y}\ket{}{z}.
\end{align}
Finally the output states are given by tracing out the initial output system 
\begin{equation}\label{Eq:gen}
\sigma_x = \frac{1}{2} \sum_{y,z,z'} (-1)^{x(z+z')} \sqrt{W(y|z) W(y|z')} \ket{}{y}\ketbra{}{z}{y}\bra{}{z'}.
\end{equation}
We denote the channel dual to $W$ as $W^{\bot}$.
In the same manner we can define dual channels for arbitrary classical-quantum channels following the steps above starting from Equation~\ref{Udual} with the $\ket{}{\varphi_z}$ being purifications of the output states of the given channel. \\

This now allows us to calculate the duals of specific channels and also for combinations of channels. One result we state in the following Lemma, which is Theorem 1 in \cite{R17}. 
\begin{lem}\label{lem:boxvaro}
Let $W_1$ and $W_2$ be two binary input cq-channels, then the following holds
\begin{align}
W_1^{\bot}\boxast W_2^{\bot} &= (W_1 \varoast W_2)^{\bot} \\
W_1^{\bot}\varoast W_2^{\bot} &= (W_1 \boxast W_2)^{\bot}.
\end{align}
\end{lem}

We want to combine above Lemma~\ref{lem:boxvaro} with an observation made in \cite{RB08,WR12a} witch states that for any $W$
\begin{equation}\label{II1}
I(W) + I(W^{\bot}) = \log 2, 
\end{equation} 
which leads us to 
\begin{equation}\label{Hbv}
H( W_1 \varoast W_2) = \log 2 - H(W_1^{\bot}\boxast W_2^{\bot}).
\end{equation}
 Note that in general $(W^{\bot})^{\bot} \neq W$~\cite{R17}, although this relation becomes an equality if $W$ is symmetric, but in either case from Equation~\ref{II1} we can directly conclude that 
\begin{equation}
H((W^{\bot})^{\bot})=H(W).
\end{equation}

From the above arguments we can directly make an important observation. 
Namely, let $W_j$ be the channels corresponding to the states $\rho^{X_jB_j}$ ($j=1,2$), which in particular means $H(W_j)=H(X_j|B_j)=H_j$. Then we have the following chain of equalities:
\begin{align}
H(X_1+X_2|B_1B_2)&=H(W_1\boxast W_2)\nonumber\\
&=H(W_1)+H(W_2)-H(W_1\varoast W_2)\nonumber\\
&=H_1+H_2-H((W_1^{\bot}\boxast W_2^{\bot})^{\bot})\nonumber\\
&=H_1+H_2-\log2+H(W_1^{\bot}\boxast W_2^{\bot})\label{eqntermafterperpequation}
\end{align}
where the first line is by definition of $\boxast$, the second line the chain rule for mutual information (conservation of entropy), the third line follows from Lemma \ref{lem:boxvaro}, and the fourth line follows from Eq.\ (\ref{II1}).

In particular this can be rewritten, using Equation~\ref{Hbv}, as
\begin{equation}
H(W_1\varoast W_2)  - \left( H(W_1)+H(W_2) \right)/2 =  H(W_1^{\bot}\varoast W_2^{\bot}) - \left( H(W_1^{\bot})+H(W_2^{\bot}) \right)/2. 
\end{equation}
This is especially interesting, because it follows directly that, due to the additional uncertainty relation given by Equation~\ref{II1}, the lower bound in the quantum setting has an additional symmetry w.r.t. the transformation $H_i\mapsto\log2-H_i$, which the classical bound does not have. Therefore one can also easily see that there must exist states with quantum side information that violate the classical bound. \\

 Finally we will give two particular examples of duals to classical channels, which were already provided in \cite{RSH14}, which state that the dual of every binary symmetric channel is a channel with pure state outputs and the dual of a BEC is again a BEC. 
\begin{example}{Binary symmetric channel (Example 3.8 in~\cite{RSH14})}
Let W be the classical BSC($p$), for every $p$ the output states of the dual channel are of the form
\begin{equation}
\sigma_x = \ketbra{}{\theta_x}{\theta_x},  \label{BSCdual}
\end{equation}
with $\ket{}{\theta_x} = Z^x \left( p\ket{}{0} + (1-p)\ket{}{1}\right)$, where $Z$ is the Pauli-Z matrix.
\end{example}
\begin{example}{Binary erasure channel (Example 3.7 in~\cite{RSH14})}
Let W be the classical BEC($p$), for every $p$ the the dual channel is again a binary erasure channel, now with erasure probability $1-p$. 
\end{example}
These examples will become useful again when discussing our conjectured optimal bound. 
 
\section{Bounds on the concavity of the von Neumann entropy}
Later we will need to relate the fidelity characteristic $f=F(\rho_0,\rho_1)$ of a binary-input classical-quantum channel with output states $\rho_0$ and $\rho_1$ back to its symmetric capacity $\log2-H$ (first, not necessarily for uniformly distributed inputs), therefore we need a lower bound on the concavity of the von Neumann entropy. This will be a special case of the following new bounds (see also Remark \ref{remarkOnRenesBound}):

\begin{thm}[Lower bounds on concavity of von Neumann entropy]\label{vNconcavityimprovement}
Let $\rho_i\in{\mathcal B}(\mathbb C^d)$ be quantum states for $i=1,\ldots,n$ and $\{p_i\}_{i=1}^n$ be a probability distribution. Then:
\begin{align}
H\left(\sum_{i=1}^np_i\rho_i\right)&-\sum_{i=1}^np_iH(\rho_i)\nonumber\\
&=H(\{p_i\})-D\Big(\sum_{i,j=1}^n\sqrt{p_ip_j}|i\rangle\langle j|\otimes\sqrt{\rho_i}\sqrt{\rho_j}\Big\|\sum_{i=1}^np_i|i\rangle\langle i|\otimes\rho_i\Big)\label{equalityinconcavity}\\
&\geq H(\{p_i\})-\log\Big(1+2\sum_{1\leq i<j\leq n}\sqrt{p_ip_j}{\rm tr}[\sqrt{\rho_i}\sqrt{\rho_j}]\Big)\label{concavityLBwithSQRT}\\
&\geq H(\{p_i\})-\log\Big(1+2\sum_{1\leq i<j\leq n}\sqrt{p_ip_j}F(\rho_i,\rho_j)\Big).\label{concavityLBwithF}
\end{align}
\end{thm}
\begin{proof}We will obtain the equality (\ref{equalityinconcavity}) by keeping track of the gap term in the proof of the upper bound on the concavity in \cite[Theorem 11.10]{NC00}, and the further inequalities by bounding the relative entropy from above. For the proof, define $\rho:=\sum_{i=1}^np_i\rho_i$.

Denote by $|\Omega\rangle_{AC}:=\sum_{i=1}^d|i\rangle_A\otimes|i\rangle_C$ the (unnormalized) maximally entangled state between two systems $A$ and $C$ of dimension $d$. Then $|\phi_i\rangle_{AC}:=(\mathbbm1\otimes\sqrt{\rho_i})|\Omega\rangle_{AC}$ are purifications of the $\rho_i$ in the sense that ${\rm tr}_A[|\phi_i\rangle\langle\phi_i|_{AC}]=\rho_i$. We also have ${\rm tr}_C[|\phi_i\rangle\langle\phi_i|_{AC}]=\rho_i^T$, where $^T$ denotes the transposition w.r.t.\ the basis $\{|i\rangle\}_A$. For a system $B$ of dimension $n$ with orthonormal basis $\{|i\rangle_B\}_{i=1}^n$, the state
\begin{align*}
|\psi\rangle_{ABC}:=\sum_{i=1}^n\sqrt{p_i}|i\rangle_B\otimes|\phi_i\rangle_{AC}
\end{align*}
is therefore a purification of $\rho^T$ in the sense that $\rho^T=\psi_A:={\rm tr}_{BC}[\psi_{ABC}]$, where we have defined $\psi_{ABC}:=|\psi\rangle\langle\psi|_{ABC}$. Since the transposition leaves the spectrum invariant, we have $S(\rho)=S(\rho^T)=S(\psi_A)=S(\psi_{BC})$, where
\begin{align*}
\psi_{BC}:={\rm tr}_A[\psi_{ABC}]=\sum_{i,j=1}^n\sqrt{p_ip_j}|i\rangle\langle j|_B\otimes(\sqrt{\rho_i}\sqrt{\rho_j})_C.
\end{align*}

Consider now the map $P_B(X):=\sum_{i=1}^n|i\rangle\langle i|_BX|i\rangle\langle i|_B$ acting on subsystem $B$, such that $P_B(\psi_{BC})=\sum_{i=1}^np_i|i\rangle\langle i|\otimes\rho_i$. Note that $P_B=P_B^*$ represents a projective measurement on $B$ and is selfadjoint w.r.t.\ the Hilbert-Schmidt inner product. We can therefore write:
\begin{align*}
D(\psi_{BC}\|P_B(\psi_{BC}))&=-H(\psi_{BC})-{\rm tr}[\psi_{BC}\log P_B(\psi_{BC})]\\
&=-H(\psi_{BC})-{\rm tr}[P_B(\psi_{BC})\log P_B(\psi_{BC})]\\
&=-H(\rho)+H(P_B(\psi_{BC}))\\
&=-H(\rho)+H(\{p_i\})+\sum_{i=1}^np_iS(\rho_i),
\end{align*}
which proves the equality (\ref{equalityinconcavity}).

To obtain the lower bound (\ref{concavityLBwithSQRT}), we bound the relative entropy from above by the \emph{sandwiched Renyi-$\alpha$ divergence} of order $\alpha=2$ \cite{Wilde2014asdf,muller2013quantum,tomamichel2015quantum}:
\begin{align*}
D(\psi_{BC}\|P_B(\psi_{BC}))&\leq D_2(\psi_{BC}\|P_B(\psi_{BC}))=\log{\rm tr}[(P_B(\psi_{BC}))^{-1/2}\psi_{BC}(P_B(\psi_{BC}))^{-1/2}\psi_{BC}].
\end{align*}
Note that the sandwiched Renyi divergences are the \emph{minimal} quantum generalizations of the classical Renyi-$\alpha$ divergences \cite{tomamichel2015quantum}, which will be advantageous to obtain a good lower bound. We can continue by using the explicit forms of $\psi_{BC}$ and $P_B(\psi_{BC})$ from above:
\begin{align*}
D(\psi_{BC}\|P_B(\psi_{BC}))&\leq\log{\rm tr}\Big[\Big(\sum_{i,j=1}^n|i\rangle\langle j|_B\otimes\mathbbm1_C\Big)\Big(\sum_{k,l=1}^n\sqrt{p_kp_l}|k\rangle\langle l|_B\otimes\sqrt{\rho_k}\sqrt{\rho_l}\Big)\Big]\\
&=\log\Big(\sum_{i,j=1}^n\sqrt{p_ip_j}{\rm tr}[\sqrt{\rho_i}\sqrt{\rho_j}]\Big),
\end{align*}
which agrees with (\ref{concavityLBwithSQRT}) since the terms with $i=j$ sum to $\sum_{i=1}^np_i{\rm tr}[\rho_i]=1$. The final bound (\ref{concavityLBwithF}) is obtained by noting that $F(\rho_i,\rho_j)=\|\sqrt{\rho_i}\sqrt{\rho_j}\|_1\geq{\rm tr}[\sqrt{\rho_i}\sqrt{\rho_j}]$ holds for any quantum states \cite{NC00,audenaert2012comparisons}.
\end{proof}

\begin{remark}[Upper bounds on concavity of von Neumann entropy]
The equality (\ref{equalityinconcavity}) in Theorem \ref{vNconcavityimprovement} can also be used to obtain \emph{upper} bounds on the concavity of von Neumann entropy: As opposed to the proof of Theorem \ref{vNconcavityimprovement}, where we used the upper bound $D\leq D_2$ involving the sandwiched Renyi-$2$ divergence, one could bound the relative entropy $D$ from below, e.g.\ using the Pinsker inequality~\cite{W13} or using a smaller divergence measure such as one of the various Renyi-$\alpha$ divergences with parameter $\alpha\in[0,1)$.
\end{remark}

We will later need the special case $n=2$ of Theorem \ref{vNconcavityimprovement} with uniform probabilities $\{p_i\}$ together with a bound from \cite{PhysRevLett.105.040505}, in order to obtain a bound on the fidelity parameter $f$ in terms of the channel entropy $H$ for binary-input classical-quantum channels.
\begin{thm}[Relation between fidelity parameter and channel entropy]\label{tightRelationFHtheorem}
Let $\sigma_0,\sigma_1\in{\mathcal B}(\mathbb C^d)$ be quantum states, and define $f:=F(\sigma_0,\sigma_1)$ and $H=\log2-H((\sigma_0+\sigma_1)/2)+(H(\sigma_0)+H(\sigma_1))/2=H(X|B)$, where $H(X|B)$ is evaluated on the state $\frac{1}{2}|0\rangle\langle0|_X\otimes(\sigma_0)_B+\frac{1}{2}|1\rangle\langle1|_X\otimes(\sigma_1)_B$. Then the following bound holds:
\begin{align}\label{tightrelationHf}
e^H-1\leq f\leq1-2h_2^{-1}(\log2-H),
\end{align}
where $h_2^{-1}:[0,\log2]\to[0,1/2]$ is the the inverse of the binary entropy function.
\end{thm}
\begin{proof}The lower bound follows immediately from Theorem \ref{vNconcavityimprovement} in the special case of $n=2$  states $\sigma_0,\sigma_1$ with equal probabilities $p_0=p_1=1/2$:
\begin{align*}
\log2-H=H\Big(\frac{\sigma_0+\sigma_1}{2}\Big)-\frac{H(\sigma_0)+H(\sigma_1)}{2}\geq\log2-\log(1+F(\sigma_0,\sigma_1))=\log2-\log(1+f).
\end{align*}
For the other direction, we need the following bound from \cite{PhysRevLett.105.040505}:
\begin{align*}
\log2-H=H\Big(\frac{\sigma_0+\sigma_1}{2}\Big)-\frac{H(\sigma_0)+H(\sigma_1)}{2}\leq h_2\Big(\frac{1-F(\sigma_0,\sigma_1)}{2}\Big)=h_2\Big(\frac{1-f}{2}\Big),
\end{align*}
where $h_2$ is the binary entropy function. The upper bound in (\ref{tightrelationHf}) follows now by noting that the inverse function $h_2^{-1}:[0,\log2 ]\to[0,1/2]$ is monotonically increasing.
\end{proof}

\begin{remark}\label{remarkonotherconcavitybounds}The main feature of the bound (\ref{tightrelationHf}) for our purposes is that it is tight on \emph{both} ends of the interval $H\in[0,\log2]$. Namely, the bound implies $H=0\Leftrightarrow f=0$ as well as $H=\log2\Leftrightarrow f=1$, see also Fig.\ \ref{figboundHf}. In particular, we are not aware of any previous bound showing that small fidelity $f\approx0$ implies $H$ to be close to $0$. Such a statement, however, is needed for our proofs of Theorems \ref{QuantumMrsGerberTheorem2DifferentStates} and \ref{QuantumMrsGerberTheorem1StateTwice} (see Eqs.\ (\ref{first-asymmetric-lower-bound-eqn}) and (\ref{minof2expressionsinproofofHH})).

In particular, the bound $\log2-H\geq\frac{1}{2}\left(\frac{1}{2}\|\sigma_0-\sigma_1\|_1\right)^2$, which is the main result of \cite{kim2014bounds}, can never yield any non-trivial information for $H\in[0,(\log2)-1/2)$ (i.e.\ near $f\approx0$), since its right-hand side will never exceed $\frac{1}{2}$. Using the Fuchs-van de Graaf inequality $\frac{1}{2}\|\sigma_0-\sigma_1\|_1\geq1-f$ \cite{NC00}, we would only obtain the bound $f\geq1-\sqrt{2(\log2-H)}$, which is also shown in Fig.\ \ref{figboundHf}.

Our lower bounds (\ref{concavityLBwithSQRT}) and (\ref{concavityLBwithF}) are generally good when the states $\rho_i$ are close to pairwise orthogonal: If $\max_{i\neq j}F(\rho_i,\rho_j)$ becomes close to $0$ then these lower bounds approach the value $H(\{p_i\})$, which is the value of the left-hand-side of the inequality for exactly pairwise orthogonal states $\rho_i$. Note however that the lower bounds (\ref{concavityLBwithSQRT}) and (\ref{concavityLBwithF}) can become negative and therefore trivial, e.g.\ when all states $\rho_i$ coincide (or have high pairwise fidelity) and the probability distribution $\{p_i\}$ is not uniform on its support. For a uniform probability distribution $\{p_i=1/n\}$ and any states $\rho_i$, the bounds (\ref{concavityLBwithSQRT}) and (\ref{concavityLBwithF}) are however always nonnegative (this case also covers Theorem \ref{tightRelationFHtheorem}).
\end{remark}
\begin{remark}\label{remarkOnAlexDaniel}
In \cite{MFW16} a different lower bound on the concavity of the von Neumann entropy was found which was shown to outperform the bound in \cite{kim2014bounds} in some cases. The bound is given in terms of the relative entropy, which can be easily bounded by $D(\rho || \sigma) \geq -2 \log{F(\rho, \sigma)}$, see \cite{muller2013quantum}. Nevertheless this bound can not be used in our general scenario since it becomes trivial whenever the involved states are pure. Note that this is not the case for our bound presented above. 
\end{remark}
\begin{remark}\label{remarkOnRenesBound}
Shortly before the initial submission of our paper we discovered that the bound from Theorem \ref{tightRelationFHtheorem} (the case of uniform input distribution) has recently been given in \cite{NR17}, also in the context of Polar codes. Our Theorem \ref{vNconcavityimprovement} is however more general, it constitutes an \emph{equality} form of the concavity of the von Neumann entropy which allows for convenient relaxations, and is valid for \emph{non-uniform} distributions. Furthermore a weaker bound can already be found in~\cite{SRDR13}.
\end{remark}

\begin{figure}[t!]
\centering
\includegraphics[trim=3.1cm 21.3cm 8.7cm 2.6cm, clip, scale=0.8]{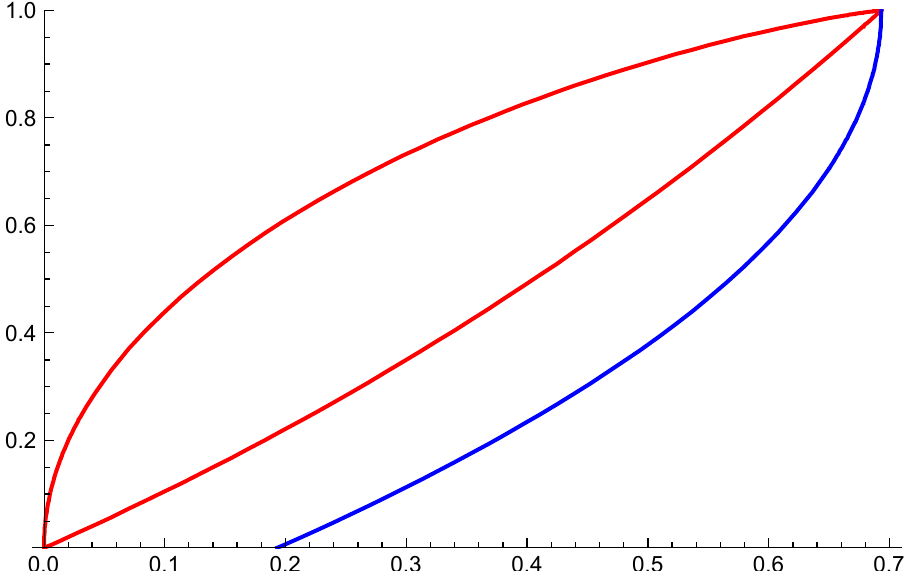}
\caption{\label{figboundHf}The $x$-axis is $H$ and the $y$-axis is $f$. The red curves show the upper and lower bounds from (\ref{tightrelationHf}), the blue curve the lower bound via \cite{kim2014bounds} (see Remark \ref{remarkonotherconcavitybounds}).}
\end{figure}

\section{Nontrivial bound for special case of Mrs.\ Gerber's Lemma with quantum side information}\label{lower-bound}

For general \emph{quantum} side information, we prove in this section nontrival lower bounds akin to the classical Mrs.\ Gerber's Lemma, albeit only for the special case when the a priori probabilities are uniform, i.e.\ $p(X_1=0)=p(X_2=0)=1/2$. This case is relevant for several applications, as we show in later sections. A conjecture of the optimal bound, also covering the case of nonuniform probabilities, is made in Section \ref{main}.

\begin{thm}[Mrs.\ Gerber's Lemma with quantum side information for uniform probabilities]\label{QuantumMrsGerberTheorem2DifferentStates}
Let $\rho^{X_1B_1}$ and $\rho^{X_2B_2}$ be \emph{independent} classical-quantum states carrying uniform a priori classical probabilites on the binary variables $X_1$, $X_2$, i.e.\
\begin{equation}
\rho^{X_jB_j}=\frac{1}{2}{\ketbra{}{0}{0}}_{X_j}\otimes\sigma_0^{B_j}+\frac{1}{2}{\ketbra{}{1}{1}}_{X_j}\otimes\sigma_1^{B_j},
\end{equation}
where $\sigma_i^{B_j}\in{\mathcal B}({\mathbb C}^{d_j})$ are quantum states on a $d_j$-dimensional Hilbert space ($i,j=1,2$). Denote their conditional entropies by $H_1 = H(X_1|B_1)$ and $H_2 = H(X_2|B_2)$, respectively. Then the following entropy inequality holds:
\begin{equation}
\begin{split}\label{firstFGmainresult}
%H(X_1+X_2|B_1B_2)&\geq H_1-2\log\cos\left[\frac{1}{2}\arccos[(1-2h_2^{-1}(\log2-H_1))(e^{H_2}-1)]-\frac{1}{2}\arccos[e^{H_2}-1]\right]\,,\\
%H(X_1+X_2|B_1B_2)&\geq H_2-2\log\cos\left[\frac{1}{2}\arccos[(1-2h_2^{-1}(\log2-H_2))(e^{H_1}-1)]-\frac{1}{2}\arccos[e^{H_1}-1]\right]\,.\\
H(X_1+X_2|B_1B_2)\geq\max \Big\{  & H_1-2\log\cos\left[\frac{1}{2}\arccos[(1-2h_2^{-1}(\log2-H_1))(e^{H_2}-1)]-\frac{1}{2}\arccos[e^{H_2}-1]\right]\,,\\
& H_2-2\log\cos\left[\frac{1}{2}\arccos[(1-2h_2^{-1}(\log2-H_2))(e^{H_1}-1)]-\frac{1}{2}\arccos[e^{H_1}-1]\right]\,,\\
& H_2-2\log\cos\left[\frac{1}{2}\arccos[(1-2h_2^{-1}(H_1))(2e^{-H_2}-1)]-\frac{1}{2}\arccos[2e^{-H_2}-1]\right]\,,\\
& H_1-2\log\cos\left[\frac{1}{2}\arccos[(1-2h_2^{-1}(H_2))(2e^{-H_1}-1)]-\frac{1}{2}\arccos[2e^{-H_1}-1]\right]\Big\}\,.
\end{split}
\end{equation}
\end{thm}
\begin{proof}We first prove that $H(X_1+X_2|B_1B_2)$ is not smaller than the first expression in the $\max$ in (\ref{firstFGmainresult}). To begin with, note the following:
\begin{align}
I(X_1+X_2:X_2|B_1B_2)&=H(X_1+X_2|B_1B_2)+H(X_2|B_1B_2)-H(X_1+X_2,X_2|B_1B_2)\nonumber\\
&=H(X_1+X_2|B_1B_2)+H(X_2|B_1B_2)-H(X_1,X_2|B_1,B_2)\nonumber\\
&=H(X_1+X_2|B_1B_2)+H(X_2|B_2)-H(X_1|B_1)-H(X_2|B_2)\nonumber\\
&=H(X_1+X_2|B_1B_2)-H_1,\label{writewithIcondmutualinfo}
\end{align}
where the first equality is just the definition, the second uses that there is a bijective (or unitary) relation between $(X_1+X_2,X_2)$ and $(X_1,X_2)$ (namely, a CNOT gate), and the third uses (twice) that $X_1B_1$ and $X_2B_2$ are independent.

Whereas the strong subadditivity property of von Neumann entropy~\cite{Lieb2002, NC00} guarantees that $I(X_1+X_2:X_2|B_1B_2)\geq0$, and therefore $H(X_1+X_2|B_1B_2)-H_1$ is nonnegative, we employ a recently established breakthrough result by Fawzi and Renner \cite{FR14}, giving a nontrivial lower bound on the quantum conditional mutual information, in order to derive our inequality (\ref{firstFGmainresult}). The result in \cite{FR14} provides a lower bound based on the so called Fidelity of Recovery. This bound has later been improved in several ways, including the structure of the involved recovery map~\cite{JRSWW15, STH16, SBT16}, stronger bounds in terms of the measured relative entropy~\cite{BHOS15}, and providing an operational interpretation~\cite{CHMOSWW16}. 

To apply the inequality from \cite{FR14}, we introduce the quantum state $\tau_{ACB}$ with binary (classical) registers $A=X_1+X_2$ and $C=X_2$, and a quantum register $B=B_1B_2$:
\begin{align}
\tau_{ACB} &\equiv \tau_{(X_1+X_2)(X_2)(B_1B_2)} := {\rm CNOT}_{(X_1,X_2)\mapsto(X_1+X_2,X_2)}\big(\rho^{X_1B_1}\otimes\rho^{X_2B_2}\big)\\
&=\frac{1}{4}{\ketbra{}{0}{0}}_A\otimes{\ketbra{}{0}{0}}_C\otimes\sigma_0^{B_1}\otimes\sigma_0^{B_2}+\frac{1}{4}{\ketbra{}{1}{1}}_A\otimes{\ketbra{}{0}{0}}_C\otimes\sigma_1^{B_1}\otimes\sigma_0^{B_2}\nonumber\\
&\quad+\frac{1}{4}{\ketbra{}{1}{1}}_A\otimes{\ketbra{}{1}{1}}_C\otimes\sigma_0^{B_1}\otimes\sigma_1^{B_2}+\frac{1}{4}{\ketbra{}{0}{0}}_A\otimes{\ketbra{}{1}{1}}_C\otimes\sigma_1^{B_1}\otimes\sigma_1^{B_2}\nonumber\\
&=\frac{1}{2}{\ketbra{}{0}{0}}_C\otimes\omega_0^{AB_1}\otimes\sigma_0^{B_2}+\frac{1}{2}{\ketbra{}{1}{1}}_C\otimes\omega_1^{AB_1}\otimes\sigma_1^{B_2},
\end{align}
where we defined in the last line $\omega_0^{AB_1}:=\frac{1}{2}({\ketbra{}{0}{0}}_A\otimes\sigma^{B_1}_0+{\ketbra{}{1}{1}}_A\otimes\sigma^{B_1}_1)$ and $\omega_1^{AB_1}:=\frac{1}{2}({\ketbra{}{0}{0}}_A\otimes\sigma^{B_1}_1+{\ketbra{}{1}{1}}_A\otimes\sigma^{B_1}_0)$ for later convenience. The main result of \cite{FR14} now says that there exists a quantum channel ${\mathcal R}'_{B\to AB}$ such that the following inequality holds:
\begin{align}
H(X_1+X_2|B_1B_2)-H_1&=I(A:C|B)_\tau\nonumber\\&\geq-2 \log F(\tau_{ACB}, \mathcal R'_{B\rightarrow AB}(\tau_{CB}))\label{fawzi-renner-bound-in-derivation} \\
&=-2\log\left[\frac{1}{2}F(\omega_0^{AB_1}\otimes\sigma_0^{B_2},{\mathcal R}'_{B\to AB}(\overline{\sigma}^{B_1}\otimes\sigma_0^{B_2}))+\frac{1}{2}F(\omega_1^{AB_1}\otimes\sigma_1^{B_2},{\mathcal R}'_{B\to AB}(\overline{\sigma}^{B_1}\otimes\sigma_1^{B_2}))\right]\nonumber  \\
&= -2\log\left[\frac{1}{2} F(\omega_0^{AB_1} \otimes \sigma_0^{B_2} , \mathcal R_{B_2\to AB}(\sigma^{B_2}_0)) + \frac{1}{2}F(\omega_1^{AB_1} \otimes \sigma_1^{B_2} , \mathcal R_{B_2\to AB}(\sigma^{B_2}_1)) \right],\label{minus2logAvgF}
\end{align}
where we introduced $\overline{\sigma}^{B_1}:=\frac{1}{2}(\sigma_0^{B_1}+\sigma_1^{B_1})$ and used in the third line that both $\tau_{ACB}$ and ${\mathcal R}'_{B\to AB}(\tau_{CB})$ are block-diagonal on the $C$-system to partially evaluate the fidelity, and in the fourth line defined the quantum channel ${\mathcal R}_{B_2\to AB}(\sigma_{B_2}):={\mathcal R}'_{B\to AB}(\overline{\sigma}^{B_1}\otimes\sigma_{B_2})$.

To obtain a nontrivial lower bound on $H(X_1+X_2|B_1B_2)-H_1$, we now derive a nontrivial upper bound on the expression in the square brackets in (\ref{minus2logAvgF}). Our derivation will involve a triangle inequality on the set of quantum states in order to ``join'' the two states ${\mathcal R}_{B_2\to AB}(\sigma_{0,1}^{B_2})$ occurring in this expression. There are various ways to turn the quantum fidelity $F$ into a metric (in particular, to satisfy the triangle inequality)~\cite{tomamichel2015quantum}, e.g.\ the \emph{geodesic distance} $A(\rho,\sigma):=\arccos F(\rho,\sigma)$~\cite{NC00}, the \emph{Bures metric} $B(\rho,\sigma):=\sqrt{1-F(\rho,\sigma)}$~\cite{B69}, or the \emph{purified distance} $P(\rho,\sigma):=\sqrt{1-F(\rho,\sigma)^2}$~\cite{GLN05,R02}. The following derivation can be done analogously with either of the three, but in the end the best bound will follow via the geodesic distance $A$, which we therefore use.

Using the concavity of the $\arccos$ function on the interval $[0,1]$ in the first step and abbreviating ${\mathcal R}:={\mathcal R}_{B_2\to AB}$, we obtain:
\begin{equation}
\begin{split}\label{concavity-triangle-monotonicity}
&\!\!\!\!\!\!\!\!\!\!\!\arccos\left[\frac{1}{2} F(\omega_0^{AB_1} \otimes \sigma_0^{B_2} , \mathcal R(\sigma^{B_2}_0)) + \frac{1}{2}F(\omega_1^{AB_1} \otimes \sigma_1^{B_2} , \mathcal R(\sigma^{B_2}_1)) \right]\\
&\geq \frac{1}{2} A(\omega_0^{AB_1} \otimes \sigma^{B_2}_0 , \mathcal R(\sigma^{B_2}_0)) +\frac{1}{2} A(\omega^{AB_1}_1 \otimes \sigma^{B_2}_1 , \mathcal R(\sigma^{B_2}_1)) \\
&\geq \frac{1}{2} A(\omega_0^{AB_1} \otimes \sigma_0^{B_2} , \omega_1^{AB_1} \otimes \sigma_1^{B_2}) - \frac{1}{2}A(\mathcal R(\sigma_0^{B_2}), \mathcal R(\sigma_1^{B_2}))\\
&\geq \frac{1}{2} \arccos[ F(\omega_0^{AB_1},\omega_1^{AB_1}) F(\sigma_0^{B_2},\sigma_1^{B_2}) ] - \frac{1}{2}A(\sigma_0^{B_2},\sigma_1^{B_2}) \\
&=\frac{1}{2}\arccos[F(\sigma_0^{B_1},\sigma_1^{B_1}) F(\sigma_0^{B_2},\sigma_1^{B_2}) ] - \frac{1}{2}A(\sigma_0^{B_2},\sigma_1^{B_2}) \\
&= \frac{1}{2} \arccos[fg]  - \frac{1}{2} \arccos g,
\end{split}
\end{equation}
where in the third line we used the triangle inequality along the path $\omega_0^{AB_1}\otimes\sigma_0^{B_2}\to{\mathcal R}(\sigma_0^{B_2})\to{\mathcal R}(\sigma_1^{B_2})\to\omega_1^{AB_1}\otimes\sigma_1^{B_2}$, in the fourth line we used that the fidelity is nondecreasing under quantum channels and multiplicative on tensor product states, and in the last two lines we evaluted and abbreviated $F(\omega_0^{AB_1},\omega_1^{AB_1})=F(\sigma_0^{B_1},\sigma_1^{B_1})=:f$ and $F(\sigma_0^{B_2},\sigma_1^{B_2})=:g$. Since the $\arccos$ function is nonincreasing in $[0,1]$, the last chain of inequalities yields an upper bound on the expression in square brackets in (\ref{minus2logAvgF}), and therefore:
\begin{align}
H(X_1+X_2|B_1B_2)-H_1\geq-2\log\cos\left[\frac{1}{2}\arccos[fg]-\frac{1}{2}\arccos g\right].\label{lowerboundWithfg}
\end{align}

As the last step, it is easy to verify that the right-hand-side of the inequality (\ref{lowerboundWithfg}) is monotonically decreasing in $f\in[0,1]$ for each fixed $g\in[0,1]$, and monotonically increasing in $g$ for each fixed $f$. Therefore, in order to continue the lower bound (\ref{lowerboundWithfg}), we can replace $f$ by an upper bound on $F(\sigma_0^{B_1},\sigma_1^{B_1})$ that is consistent with the given value of $H_1=H(X_1|B_1)$; and similarly replace $g$ by a lower bound on $F(\sigma_0^{B_2},\sigma_1^{B_2})$ consistent with $H_2=H(X_2|B_2)$. Exactly such upper and lower bounds are given in Theorem \ref{tightRelationFHtheorem}, following from bounds on the concavity of the von Neumann entropy, and result in
\begin{align}\label{first-asymmetric-lower-bound-eqn}
H(X_1+X_2|B_1B_2)-H_1\geq-2\log\cos\left[\frac{1}{2}\arccos[(1-2h_2^{-1}(\log2-H_1))(e^{H_2}-1)]-\frac{1}{2}\arccos[e^{H_2}-1]\right]\,,
\end{align}
showing that the first expression in the $\max$ in (\ref{firstFGmainresult}) is indeed a lower bound on $H(X_1+X_2|B_1B_2)$.

The same reasoning with $\rho^{X_1B_1}$ and $\rho^{X_2B_2}$ interchanged shows that the second expression in the $\max$ in (\ref{firstFGmainresult}) is a lower bound on $H(X_1+X_2|B_1B_2)$ as well.

To show that the third expression in the $\max$ in (\ref{firstFGmainresult}) is a lower bound on $H(X_1+X_2|B_1B_2)$, we exploit the symmetries of binary input classical-quantum channels and their complementary channels under the polar transform. For this, we recall from Section~\ref{duality} that 
\begin{equation}\label{eqntermafterperpequationrepeat}
H(X_1+X_2|B_1B_2)=H(W_1\boxast W_2)=H_1+H_2-\log2+H(W_1^{\bot}\boxast W_2^{\bot}).
\end{equation}
Thus, we can obtain another lower bound on $H(X_1+X_2|B_1B_2)$ by bounding the term $H(W_1^{\bot}\boxast W_2^{\bot})$ from below using the first expression in the $\max$ in (\ref{firstFGmainresult}); this gives the following lower bound on $H(X_1+X_2|B_1B_2)$:
\begin{align}
\ldots\geq H_1+H_2-\log2+H(W_1^{\bot})-2\log\cos\left[\frac{1}{2}\arccos[(1-2h_2^{-1}(\log2-H(W_1^{\bot}))(e^{H(W_2^{\bot})}-1)]-\frac{1}{2}\arccos[e^{H(W_2^{\bot})}-1]\right]\,,\nonumber
\end{align}
which is exactly the third expression in the $\max$ in (\ref{firstFGmainresult}) as, again by (\ref{II1}), the channels $W_j^{\bot}$ satisfy $H(W_j^{\bot})=\log2-H_j$. That the fourth expression in the $\max$ in (\ref{firstFGmainresult}) is a lower bound on $H(X_1+X_2|B_1B_2)$ is obtained similarly from (\ref{eqntermafterperpequationrepeat}) by bounding the term $H(W_1^{\bot}\boxast W_2^{\bot})$ from below using the second expression in the $\max$ in (\ref{firstFGmainresult}).
\end{proof} 

\begin{remark}\label{equality-condition-for-proof} Since the $\arccos$ function is stricly monotonically decreasing in $[0,1]$, one can see from the first expression in the $\max$ in (\ref{firstFGmainresult}) (cf.\ also (\ref{lowerboundWithfg})) that $H(X_1+X_2|B_1B_2)=H_1$ is possible only if $H_1=\log2$ or $H_2=0$. Conversely, if $H_1=\log2$ or $H_2=0$ then actually $H(X_1+X_2|B_1B_2)=H_1$ since: {\it(a)} $H_1\leq H(X_1+X_2|B_1B_2)$ holds due to (\ref{writewithIcondmutualinfo}) along with strong subadditivity; {\it(b)} $H(X_1+X_2|B_1B_2)\leq\log2$ holds as $X_1+X_2$ is a binary register; {\it(c)} and we have, similarly to (\ref{writewithIcondmutualinfo}) together with the fact that the conditional entropy $H(X_2|X_1+X_2,B_1B_2)$ of a classical system is nonnegative:
\begin{align*}
H(X_1+X_2|B_1B_2)&\leq H(X_1+X_2|B_1B_2)_\tau+H(X_2|X_1+X_2,B_1B_2)_\tau\\
&=H(X_1+X_2,X_2|B_1B_2)_\tau\\
&=H(X_1,X_2|B_1,B_2)\\
&=H(X_1|B_1)+H(X_2|B_2)=H_1+H_2.
\end{align*}
Analogously, $H(X_1+X_2|B_1B_2)=H_2$ if and only if $H_1=0$ or $H_2=\log2$. Thus, the inequality $H(X_1+X_2|B_2B_2)\geq\max\{H_1,H_2\}$ holds with equality if and only if $H_1\in\{0,\log2\}$ or $H_2\in\{0,\log2\}$. Therefore, the inequality $H(X_1+X_2|B_1B_2)\geq(H_1+H_2)/2$ holds with equality if and only if $H_1=H_2\in\{0,\log2\}$.
\end{remark}

The lower bound (\ref{firstFGmainresult}) from Theorem \ref{QuantumMrsGerberTheorem2DifferentStates} is illustrated in Fig.\ \ref{fig-bound-with-h1h2}.

\begin{figure}[t!]
\centering
\includegraphics[trim=0.2cm 0.1cm 0.2cm 0.0cm, clip, scale=0.43]{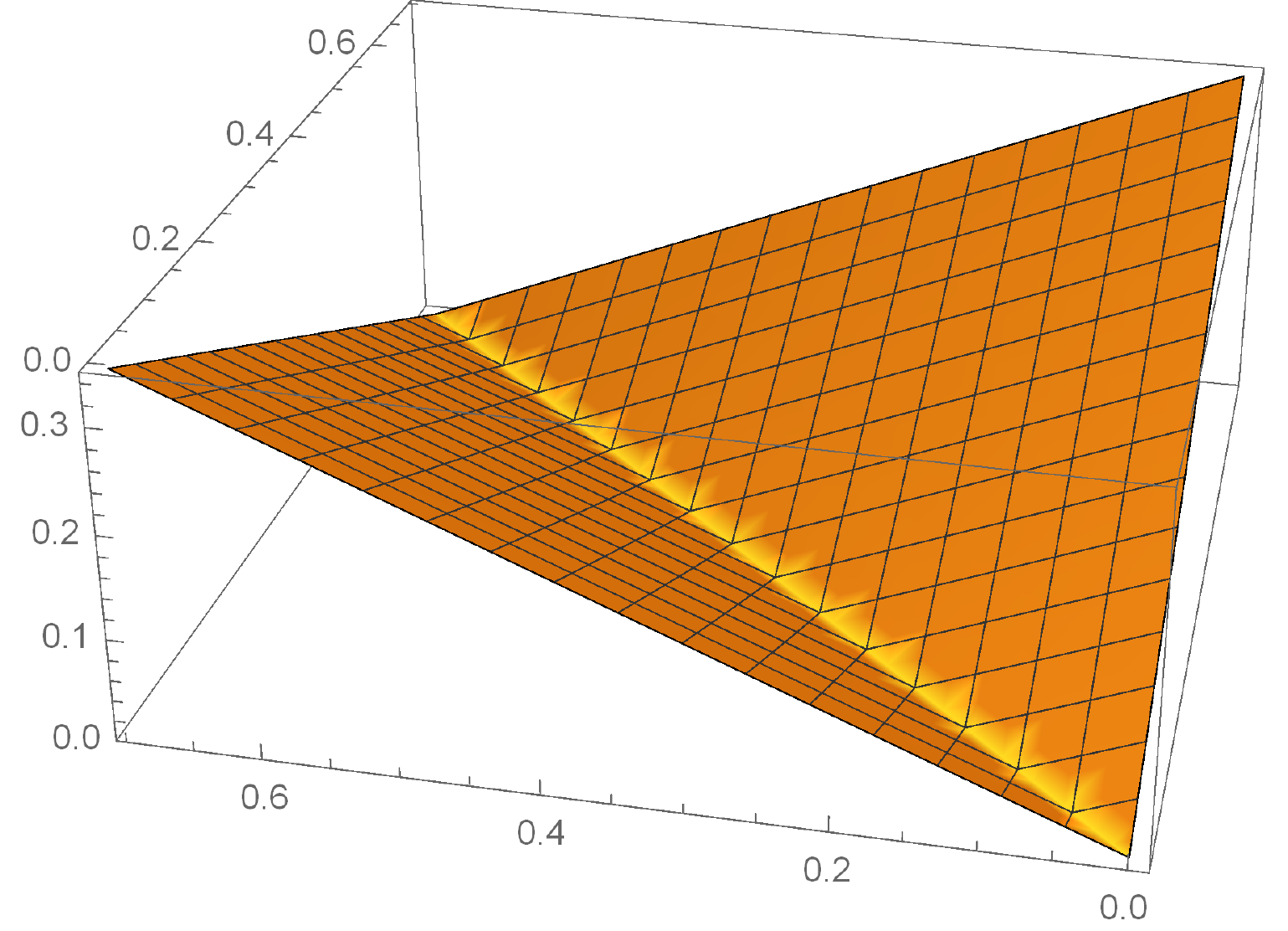}\hspace{1.5cm}\includegraphics[trim=0.0cm 0.3cm 0.0cm 0.1cm, clip, scale=0.44]{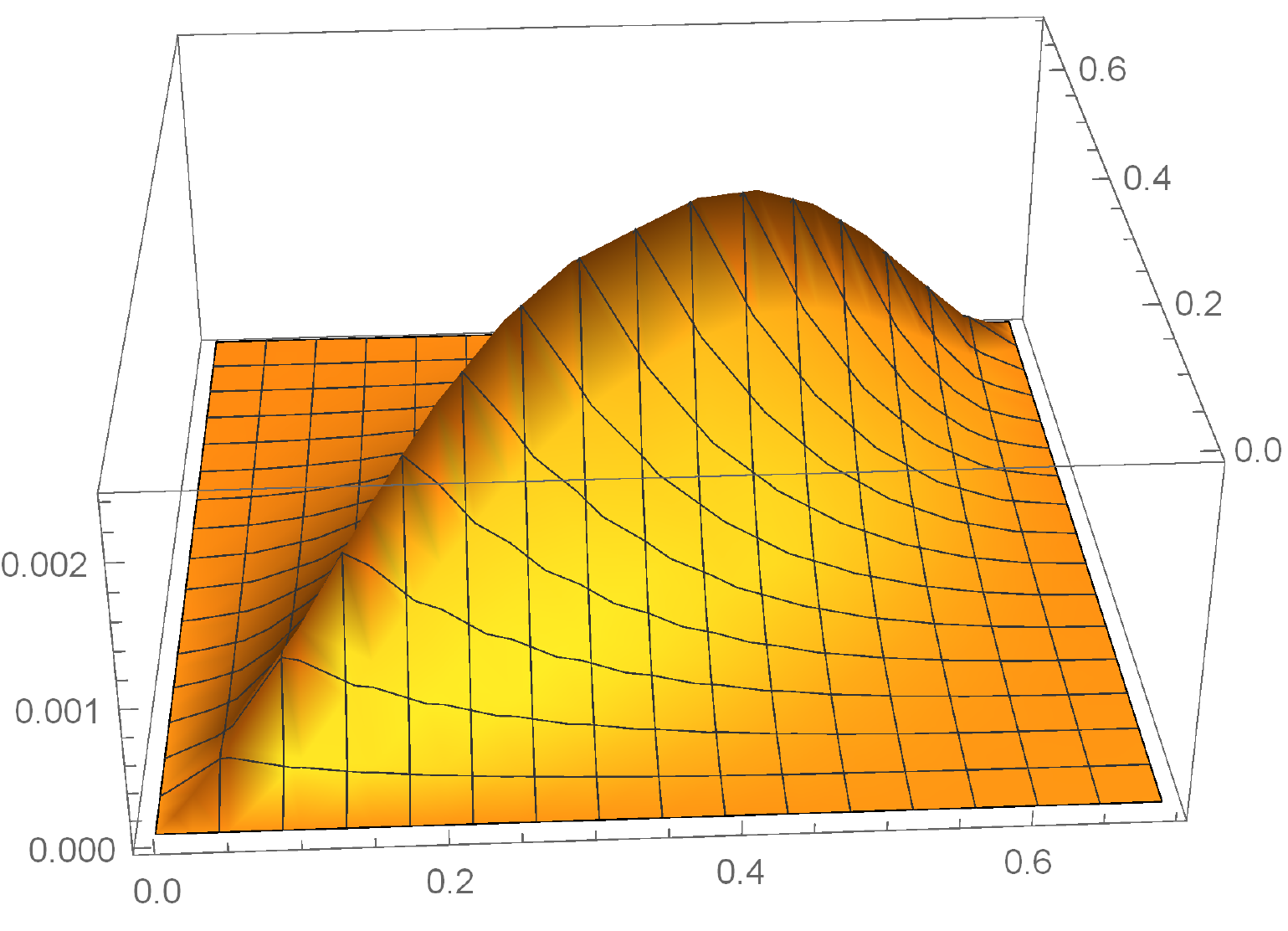}

\caption{\label{fig-bound-with-h1h2}The left plot shows the lower bound on $H(X_1+X_2|B_1B_2)-\frac{1}{2}(H_1+H_2)$ inferred from the bound (\ref{firstFGmainresult}) in Theorem \ref{QuantumMrsGerberTheorem2DifferentStates}, as a function of $(H_1,H_2)\in[0,\log2]\times[0,\log2]$; the right plot shows the lower bound on $H(X_1+X_2|B_1B_2)-\max\{H_1,H_2\}$ inferred from (\ref{firstFGmainresult}). The value of the bound along the diagonal line $H_1=H_2$ is shown again as the purple curve in Fig.\ \ref{fig-bound-with-h}.}
\end{figure}

In the important special case $\rho^{X_1B_1}=\rho^{X_2B_2}$ of Theorem \ref{QuantumMrsGerberTheorem2DifferentStates}, which will be useful e.g.\ for polar codes on i.i.d.\ channels, we can use the same proof idea to obtain a better bound:
\begin{thm}[Mrs.\ Gerber's Lemma with quantum side information on i.i.d.\ states for uniform probabilities]\label{QuantumMrsGerberTheorem1StateTwice}
Let $\rho^{X_1B_1}=\rho^{X_2B_2}$ be \emph{identical and independent} classical-quantum states carrying uniform a priori classical probabilites on the binary variables $X_1,X_2$, i.e.\
\begin{equation}
\rho^{X_1B_1}=\rho^{X_2B_2}=\frac{1}{2}{\ketbra{}{0}{0}}\otimes\sigma_0+\frac{1}{2}{\ketbra{}{1}{1}}\otimes\sigma_1,
\end{equation}
where $\sigma_i^{B_1}=\sigma_i^{B_2}=\sigma_i\in{\mathcal B}({\mathbb C}^{d})$ are quantum states on a $d$-dimensional Hilbert space ($i=1,2$). Denoting their conditional entropy by $H = H(X_1|B_1)=H(X_2|B_2)$, the following entropy inequality holds:
\begin{align}\label{FonlyMainResult}
H(X_1+X_2|B_1B_2)&\geq\left\{\begin{array}{ll}H-2\log\cos\Big[\frac{1}{2}\arccos[(1-2h_2^{-1}(H))^2]-\frac{1}{2}\arccos[1-2h_2^{-1}(H)]\Big],&H\leq\frac{1}{2}\log2\\H-2\log\cos\Big[\frac{1}{2}\arccos[(1-2h_2^{-1}(\log2-H))^2]-\frac{1}{2}\arccos[1-2h_2^{-1}(\log2-H)]\Big],&H>\frac{1}{2}\log2\end{array}\right.\\
&\geq\left\{\begin{array}{ll}H+0.083\cdot\frac{H}{1-\log H},&H\leq\frac{1}{2}\log2\\H+0.083\cdot\frac{\log2-H}{1-\log(\log2-H)},&H>\frac{1}{2}\log2.\end{array}\right.\label{FonlyMoreConvenientLowerBound}
\end{align}
The expressions (\ref{FonlyMoreConvenientLowerBound}) assume $\log$ to be the natural logarithm.
\end{thm}
\begin{proof}We follow the proof of Theorem \ref{QuantumMrsGerberTheorem2DifferentStates} up until Eq.\ (\ref{lowerboundWithfg}), which now reads
\begin{align}
H(X_1+X_2|B_1B_2)-H\geq-2\log\cos\left[\frac{1}{2}\arccos[f^2]-\frac{1}{2}\arccos f\right]\label{lowerboundWithOnlyf}
\end{align}
with $f:=F(\sigma_0,\sigma_1)$. The right-hand-side of the last lower bound is monotonically increasing for $f\in[0,1/\sqrt{3}]$ and monotonically decreasing for $f\in[1/\sqrt{3},\log2]$ since these statements hold for the function $f\mapsto\frac{1}{2}\arccos[f^2]-\frac{1}{2}\arccos f$. Therefore, a lower bound based on $e^H-1\leq f\leq1-2h_2^{-1}(\log2-H)$ from Theorem \ref{tightRelationFHtheorem} can be obtained by evaluating (\ref{lowerboundWithOnlyf}) at those boundaries:
\begin{equation}
\begin{split}\label{minof2expressionsinproofofHH}
H(X_1+X_2|B_1B_2)-H\geq\min\Big\{&-2\log\cos\Big[\frac{1}{2}\arccos[(e^H-1)^2]-\frac{1}{2}\arccos[e^H-1]\Big],\\
&-2\log\cos\Big[\frac{1}{2}\arccos[(1-2h_2^{-1}(\log2-H))^2]-\frac{1}{2}\arccos[1-2h_2^{-1}(\log2-H)]\Big]\Big\}.
\end{split}
\end{equation}
Numerically, one sees that for $H\in[\frac{1}{2}\log2,\log2]$ (and even for $H\in[0.33,\log2]$), the minimum in the last expression is attained by the second term, which gives
\begin{align}
H(X_1+X_2|B_1B_2)\geq H-2\log\cos\Big[\frac{1}{2}\arccos[(1-2h_2^{-1}(\log2-H))^2]-\frac{1}{2}\arccos[1-2h_2^{-1}(\log2-H)]\Big]\label{proofoffirstselectorinequalHcase}
\end{align}
for $H\geq\frac{1}{2}\log2$, and shows the second selector in (\ref{FonlyMainResult}). Analytically, one can easily show this statement for $H\in[\log(1+1/\sqrt{3}),\log2]$, as this implies by Theorem \ref{tightRelationFHtheorem} that $f$ is in the range $f\in[1/\sqrt{3},1]$, where the function $f\mapsto\frac{1}{2}\arccos[f^2]-\frac{1}{2}\arccos f$ is monotonically decreasing and we have $e^H-1\leq1-2h_2^{-1}(\log2-H)$ by Theorem \ref{tightRelationFHtheorem}. The statement is also true for $H\in[0.33,\log(1+1/\sqrt{3})]$, for the following reason: First, the statement is easily numerically certified for $H=0.33$; second, the function that maps $H$ to the first expression in the minimum in (\ref{minof2expressionsinproofofHH}) is monotonically \emph{in}creasing for $H\in[0,\log(1+1/\sqrt{3})]$ since $H\mapsto e^H-1$ is increasing from $0$ to $1/\sqrt{3}$, where the right-hand-side of (\ref{lowerboundWithOnlyf}) is increasing in $f$; third, the function that maps $H$ to the second expression in the minimum in (\ref{minof2expressionsinproofofHH}) is monotonically \emph{de}creasing for $H\in[0.33,\log(1+1/\sqrt{3})]$ since the function $H\mapsto1-2h_2^{-1}(\log2-H)$ is increasing and not smaller than $1-2h_2^{-1}(\log2-0.33)\geq0.76\geq1/\sqrt{3}$, where the right-hand-side of (\ref{lowerboundWithOnlyf}) is decreasing in $f$.

To prove the first selector in (\ref{FonlyMainResult}), i.e.\ the case $H\leq\frac{1}{2}\log2$, we again use the reasoning via complementary channels as in the proof of Theorem \ref{QuantumMrsGerberTheorem2DifferentStates}. Eq.\ (\ref{eqntermafterperpequation}) now reads:
\begin{align}
H(X_1+X_2|B_1B_2)=2H-\log2+H(W^{\bot}\varoast W^{\bot}),\label{mirrorsymmetryHHeqn}
\end{align}
where $W$ is the channel corresponding to the state $\rho^{X_1B_1}=\rho^{X_2B_2}$ and $W^{\bot}$ its complementary. Since $H(W^{\bot})=\log2-H\geq\frac{1}{2}\log2$ we can apply (\ref{proofoffirstselectorinequalHcase}) to the channel $W^{\bot}$ to bound the last expression from below:
\begin{align}\nonumber
\ldots\geq2H-\log2+H(W^{\bot})-2\log\cos\Big[\frac{1}{2}\arccos[(1-2h_2^{-1}(\log2-H(W^{\bot})))^2]-\frac{1}{2}\arccos[1-2h_2^{-1}(\log2-H(W^{\bot}))]\Big],
\end{align}
which with $H(W^{\bot})=\log2-H$ gives finally the desired expression in the first selector in (\ref{FonlyMainResult}).

We show the more convenient lower bound (\ref{FonlyMoreConvenientLowerBound}) by using a few inequalities without formal proof. First we employ
\begin{align}\nonumber
\frac{1}{2}\arccos[x^2]-\frac{1}{2}\arccos x\geq\frac{\frac{1}{2}\arccos[F^2]-\frac{1}{2}\arccos F}{\sqrt{1-F}}\sqrt{1-x}\qquad\forall x\in[F,1]
\end{align}
for $F:=1-2h_2^{-1}(\frac{1}{2}\log2)$, since the function $x\mapsto(\arccos[x^2]-\arccos[x])/\sqrt{1-x}$ is monotonically increasing in $x\in[0,1)$. Using this in the first selector in (\ref{FonlyMainResult}), i.e.\ for $x=1-2h_2^{-1}(H)$, we obtain for any $H\leq\frac{1}{2}\log2$:
\begin{align}
H(X_1+X_2|B_1B_2)&\geq H-2\log\cos\left[c_1\sqrt{2h_2^{-1}(H)}\right]\qquad\text{where}~c_1:=\left.\frac{\frac{1}{2}\arccos[F^2]-\frac{1}{2}\arccos F}{\sqrt{1-F}}\right|_{F=1-2h_2^{-1}(\frac{1}{2}\log2)}\nonumber\\
&\geq H-2\log\left(1-c_2c_1^2\cdot2h_2^{-1}(H)\right)\qquad\text{where}~c_2:=\left.\frac{1-\cos x}{x^2}\right|_{x=c_1\sqrt{2h_2^{-1}(\frac{1}{2}\log2)}},\nonumber
\end{align}
similar as before, since the function $x\mapsto(1-\cos x)/x^2$ is monotonically decreasing in $x\in[0,\pi/2]\ni c_1\sqrt{2h_2^{-1}(\frac{1}{2}\log2)}$. From there we continue by first using the concavity of the $\log$ function:
\begin{align}
H(X_1+X_2|B_1B_2)&\geq H+4c_2c_1^2\,h_2^{-1}(H)\nonumber\\
&\geq H+4c_2c_1^2(1-e^{-1})\frac{H}{1-\log H},\nonumber
\end{align}
where in the last step we employ a convenient lower bound on $h_2^{-1}$, containing Euler's number $e$. The first selector now follows by $4c_2c_1^2(1-e^{-1})\geq0.083$, and the second selector in (\ref{FonlyMoreConvenientLowerBound}) by interchanging $H$ and $\log2-H$.
\end{proof}

The lower bounds (\ref{FonlyMainResult}) and (\ref{FonlyMoreConvenientLowerBound}) from Theorem \ref{QuantumMrsGerberTheorem1StateTwice} are shown in Fig.\ \ref{fig-bound-with-h}, where they are also compared to the bound (\ref{firstFGmainresult}) that is obtained from Theorem \ref{QuantumMrsGerberTheorem2DifferentStates} in the case $H_1=H_2=H$.

\begin{figure}[t!]
\centering
\includegraphics[trim=3.1cm 21.8cm 8.7cm 2.6cm, clip, scale=0.8]{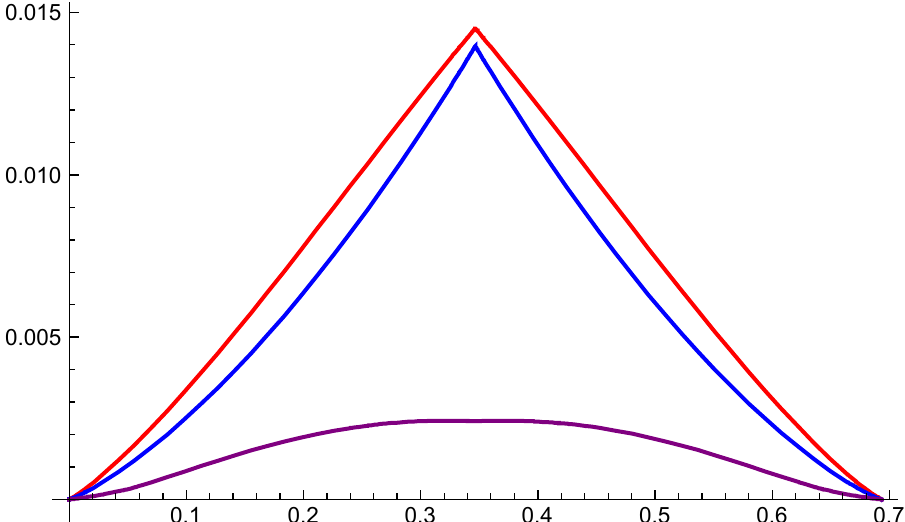}
\caption{\label{fig-bound-with-h}The red curve shows the lower bound on $H(X_1+X_2|B_1B_2)-H$ in terms of $H\in[0,\log2]$ from Eq.\ (\ref{FonlyMainResult}), the blue curve from Eq.\ (\ref{FonlyMoreConvenientLowerBound}), and the purple curve from Eq.\ (\ref{firstFGmainresult}) in the special case $H_1=H_2=H$.}
\end{figure}

\begin{remark}[Potential improvements of lower bounds] Our best lower bound (\ref{FonlyMainResult}) on $H(X_1+X_2|B_1B_2)-H$ behaves, by expanding the right-hand-side of (\ref{lowerboundWithOnlyf}) for $f\to1$, near the boundary $H\to\log2$ like
\begin{align}\label{behaviourofarccosf-expression}
\left.\left(\frac{3}{2}-\sqrt{2}\right)(1-f)+O((1-f)^2)\right|_{f=1-2h_2^{-1}(\log2-H)}=(3-\sqrt{8})h_2^{-1}(\log2-H)+O((h_2^{-1}(\log2-H))^2),
\end{align}
and thus vanishes faster than linearly $\propto(\log2-H)$ as $H\to\log2$ (see also Fig.\ \ref{fig-bound-with-h}; the behaviour for $H\to0$ follows by mirror symmetry $H\leftrightarrow\log2-H$ about $H=\frac{1}{2}\log2$). On the other hand, the bound vanishes at most as fast as $\Omega((\log2-H)/(-\log(\log2-H)))$ according to (\ref{FonlyMoreConvenientLowerBound}).

In contrast to this, our conjectured optimal lower bound from Conjecture \ref{QMGL} below posits that $H(X_1+X_2|B_1B_2)-H$ does \emph{not} vanish faster than the linear behaviour $(\log2-H)+o(\log2-H)$ for $H\to\log2$. When $\sigma_0,\sigma_1$ from Theorem \ref{QuantumMrsGerberTheorem1StateTwice} are pure states then $H=\log2-h_2((1-f)/2)$ with $f=F(\sigma_0,\sigma_1)$, and one can easily compute $H(X_1+X_2|B_1B_2)-H=h_2((1-f)/2)-(1-f)\log2+O((1-f)^2\log(1-f))=(\log2-H)+o(\log2-H)$ for $H\to\log2$ (see also Section \ref{main} for our conjectured optimal states).

If one would like to prove such a linear lower bound $\Omega(\log2-H)$ on $H(X_1+X_2|B_1B_2)-H$ for $H\to\log2$ by our proof strategy generally, one would have to improve the lower bound (\ref{lowerboundWithOnlyf}) near $f\to1$ from the linear behaviour $\Omega(1-f)$ (see (\ref{behaviourofarccosf-expression})) by a logarithmic factor, e.g.\ improve it to $\Omega(-(1-f)\log(1-f))=\Omega(h_2((1-f)/2))$ (which matches the behaviour in the pure state case described in the previous paragraph). In this respect, note that the upper bound $f\leq1-2h_2^{-1}(\log2-H)$, which is also used in our derivation (by Theorem \ref{tightRelationFHtheorem}), \emph{cannot} be improved since it is tight in the pure state case.

It is unlikely that the ``missing'' logarithmic factor in the desired $\Omega(-(1-f)\log(1-f))$ bound on the right-hand-side of (\ref{lowerboundWithOnlyf}) near $f\to1$ is due to the use of concavity, triangle inequality, and monotonicity in the part (\ref{concavity-triangle-monotonicity}) of our derivation. Rather, it is the crucial Fawzi-Renner bound itself \cite{FR14} that we use in step (\ref{fawzi-renner-bound-in-derivation}) which does not seem to be strong enough: To support this statement, we evaluate the inequality (\ref{fawzi-renner-bound-in-derivation}) again in the special setting of Theorem \ref{QuantumMrsGerberTheorem1StateTwice} (i.e.\ $\sigma_0^{B_1}=\sigma_0^{B_2}=\sigma_0$ and $\sigma_1^{B_1}=\sigma_1^{B_2}=\sigma_1$) with pure states $\sigma_0,\sigma_1$ with fidelity $f=F(\sigma_0,\sigma_1)$; and even under the optimistic assumption that the so-called \emph{Petz recovery map} ${\mathcal R}'^{Petz}$ \cite{FR14} applied in an optimal way would give a valid lower bound (which is not known to be true, and thus marked with `?' in the following), we would only obtain the following lower bound instead of (\ref{fawzi-renner-bound-in-derivation}):
\begin{align}
H(X_1+X_2|B_1B_2)-H&=I(A:C|B)_\tau\nonumber\\
&\stackrel{?}{\geq}\max\left\{-2\log F(\tau_{ACB},{\mathcal R}'^{Petz}_{B\to AB}(\tau_{CB})),-2\log F(\tau_{ACB},{\mathcal R}'^{Petz}_{B\to BC}(\tau_{AB}))\right\}\nonumber\\
&=\max\left\{-\log\left[\frac{1}{2}\left(1+f^4+(1-f^2)\sqrt{1+f^2}\right)\right],-\log\left[\frac{1}{2}\left(1+f^2+(1-f^2)^{3/2}\right)\right]\right\}\nonumber\\
&=-\log\left[\frac{1}{2}\left(1+f^2+(1-f^2)^{3/2}\right)\right]\nonumber\\
&=(1-f)+O((1-f)^{3/2})\qquad\text{as}~f\to1,\nonumber
\end{align}
which is again linear $O(1-f)$ and thus not $\Omega(-(1-f)\log(1-f))$, even though there is only one (optimistically assumed Fawzi-Renner-type) inequality in this computation.

One may hope that the desired logarithmic factor may come into a bound $\Omega(-(1-f)\log(1-f))$ improving (\ref{lowerboundWithOnlyf}) by use of recovery results employing the \emph{measured relative entropy} $D_{\mathbb M}\geq F$ instead of the fidelity~\cite{BHOS15, STH16, JRSWW15, SBT16}. We leave this possibility for further investigation. In particular, the convex programming formulations of the measured relative entropy of recovery in~\cite{BFT15} may prove useful for this purpose. During such a derivation, one may need to keep more information about the involved states $\sigma_0,\sigma_1$ than their fidelity $f=F(\sigma_0,\sigma_1)$.

As a last note, we remark that, instead of exploiting (\ref{minof2expressionsinproofofHH}) in the regime of large $H\in[\frac{1}{2}\log2,\log2]$ and afterwards symmetrizing the bound into the regime $H\in[0,\frac{1}{2}\log2]$ via (\ref{mirrorsymmetryHHeqn}), one could instead have exploited (\ref{minof2expressionsinproofofHH}) in the regime of small $H\in[0,\frac{1}{2}\log2]$ and later symmetrized towards large $H$. Any bounds that can be obtained in this way will, however, never be better than \emph{quadratic} at the boundaries, i.e.\ they will behave like $O(H^2)$ for $H\to0$ and thus $O((\log2-H)^2)$ for $H\to\log2$, and will therefore be inferior to (\ref{FonlyMainResult}) and (\ref{FonlyMoreConvenientLowerBound}) at the boundaries. The reason for this is that: {\it(a)} no lower bound in terms of the fidelity $f$ akin to (\ref{lowerboundWithOnlyf}) can be better than $f^2/2+O(f^4)$ near $f\to0$, because this is the behaviour of $H(X_1+X_2|B_1B_2)-H$ in the pure state case described above; {\it(b)} no lower bound on $f$ can be larger than linear in $H$ for $H\to0$ (such as, e.g., the desired $f\geq\Omega(\sqrt{H})$), because the (mixed) states $\sigma_0={\rm diag}(f,1-f,0)$, $\sigma_1={\rm diag}(f,0,1-f)$ satisfy the linear relation $f=F(\sigma_0,\sigma_1)=H/\log2$.
\end{remark}

\section{Conjectures for optimal bounds}\label{main}

In this section we will present conjectures on what the optimal bounds for information combining with quantum side information might be, i.e.\ the generalization of the inequalities in Eq.\ (\ref{minus-bounds}) to the case of quantum side information.

First we give a conjecture for a lower bound in analogy to the Mrs. Gerbers Lemma (compare to the left inequality in Eq.\ (\ref{minus-bounds})):
\begin{conjecture}{[Quantum Mrs. Gerber's Lemma]}\label{QMGL}
Let $\rho^{X_1B_1}$ and $\rho^{X_2B_2}$ be classical quantum states with conditional entropy $H_1 = H(X_1|B_1)$ and $H_2 = H(X_2|B_2)$ respectively. Then the following entropy inequality holds: 
\begin{equation}\label{conj:lower}
H(X_1+X_2|B_1B_2) \geq 
\begin{cases}
h(h^{-1}(H_1)\ast h^{-1}(H_2)) &H_1+H_2 \leq \log2 \\
H_1 + H_2 -\log2 + h(h^{-1}(\log2 - H_1)\ast h^{-1}(\log2 - H_2))  &H_1+H_2 \geq \log2
\end{cases}
\end{equation}
\end{conjecture}
Additionally, we conjecture the following upper bound (compare to the right inequality in Eq.\ (\ref{minus-bounds})):
\begin{conjecture}[Upper bound]\label{upper}
Let $\rho^{X_1B_1}$ and $\rho^{X_2B_2}$ be classical quantum states with conditional entropy $H_1 = H(X_1|B_1)$ and $H_2 = H(X_2|B_2)$ respectively. Then the following entropy inequality holds: 
\begin{equation}
H(X_1+X_2|B_1B_2) \leq \log2 - \frac{(\log2 - H_1)(\log2 - H_2)}{\log2}.
\end{equation}
\end{conjecture}
In the following we will discuss several observations that give strong evidence in favour of our conjectures.

\begin{figure}[t!]
\centering
\includegraphics[clip, scale=0.55]{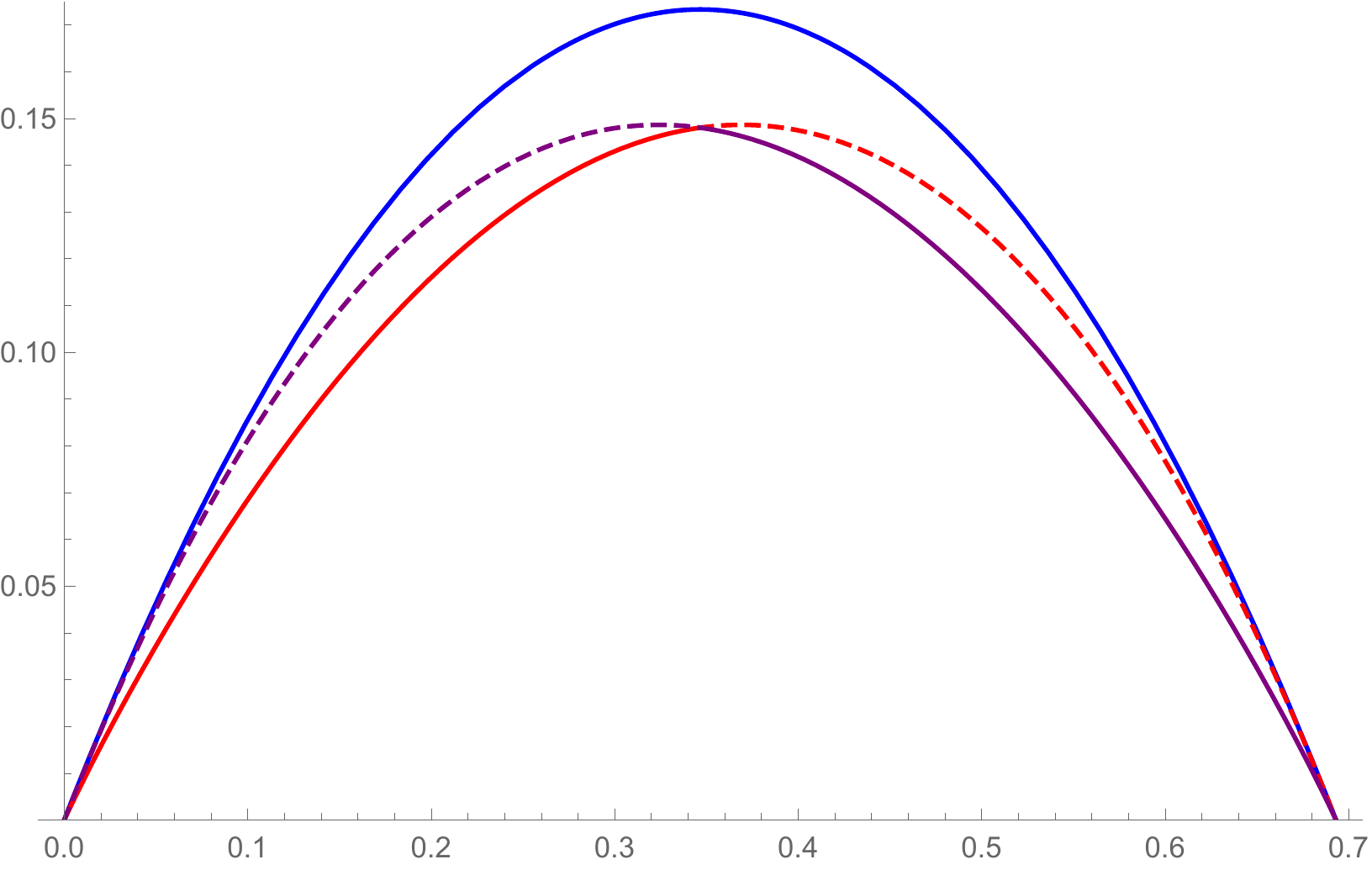}
\caption{\label{fig:conjectures} This plot shows our conjectured bounds on $H(X_1+X_2|B_1B_2) - H$ when $H_1=H_2=H$. The blue curve is the upper bound in Conjecture~\ref{upper}, while the red curve gives the lower bound for $H\leq \frac{\log 2}{2}$ and purple for $H\geq \frac{\log 2}{2}$ in Conjecture~\ref{QMGL}. Plain lines give the actual bounds, while dashed lines are shown to illustrate the two functions in Equation~\ref{conj:lower} and for comparison to the classical bound.}
\end{figure}

\subsection{Quantum states that achieve equality}\label{subsec:optimal}
First we will discuss the states that achieve equality in the conjectured inequalities. 
It can easily be seen that the classical half (i.e.\ the first selector in Eq.\ (\ref{conj:lower})) of Conjecture \ref{QMGL} can be achieved by embedding a BSC into a classical quantum state as follows (with $p\in[0,1]$ chosen accordingly):
\begin{equation}
\rho = \frac{1}{2} \ketbra{}{0}{0} \otimes (p\ketbra{}{0}{0} + (1-p)\ketbra{}{1}{1}) + \frac{1}{2} \ketbra{}{1}{1} \otimes ((1-p)\ketbra{}{0}{0} + p\ketbra{}{1}{1}).
\end{equation}
Optimality of these states follows from the inequality in the classical Mrs. Gerber's Lemma (and can also be verified easily by calculating the entropy terms).
Possibly more interesting is the quantum half of Conjecture \ref{QMGL}. The optimal states represent binary classical-quantum channels with pure output states and can therefore be represented as
\begin{equation}
\rho = \frac{1}{2} \ketbra{}{0}{0} \otimes \ketbra{}{\Psi_0}{\Psi_0}  + \frac{1}{2} \ketbra{}{1}{1} \otimes \ketbra{}{\Psi_1}{\Psi_1},
\end{equation}
where $\Psi_0$ and $\Psi_1$ are pure states. Due to unitary invariance we can choose them to be $\ket{}{\Psi_0}=\left(\begin{smallmatrix} 1 \\ 0 \end{smallmatrix}\right)$ and $\ket{}{\Psi_1}=\left(\begin{smallmatrix} \cos{\alpha } \\ \sin{\alpha }\end{smallmatrix}\right)$. 
Again this can be verified by simply calculating the involved entropies. Unfortunately this calculation is not very insightful, therefore we choose to give an alternative proof, which might also give some intuition towards why our conjectured lower bound has the given additional symmetry. 
The alternative proof will be based on the concept of dual channels as explained in Section~\ref{duality}. 

Lets fix $W_1$ and $W_2$ to be channels with pure output states of the form in Equation \ref{BSCdual} and therefore dual channels of BSCs. With the above arguments we can now show in an intuitive way that channels of this form achieve equality for the quantum side of our conjecture. 
\begin{align}
H( W_1 \boxast W_2) &= H(W_1) + H(W_2) - H( W_1 \varoast W_2) \\
&=  H(W_1) + H(W_2) -\log2 + H(W_1^{\bot}\boxast W_2^{\bot}) \\ 
&=  H(W_1) + H(W_2) -\log2 + h(h^{-1}(H(W_1^{\bot}))\ast h^{-1}(H(W_2^{\bot}))) \\
&=  H(W_1) + H(W_2) -\log2 + h(h^{-1}(\log2 - H(W_1))\ast h^{-1}(\log2 - H(W_2))), 
\end{align}
where the first inequality follows from the chain rule for mutual information, the second one from Equation \ref{Hbv}, the third from the classical Mrs. Gerbers Lemma and the final one from Equation \ref{II1}. Note that the equality holds because in the classical Mrs. Gerbers Lemma binary symmetric channels achieve equality. 

\begin{remark} With an argument along the same lines one can prove immediately that our conjectured lower bound is true not only for all states that are classical channels (or embeddings of such) but also for all states that are duals of such classical channels.  
\end{remark}

Now, lets look at Conjecture \ref{upper}. From the classical upper bound it can be easily seen that equality is achieved by embeddings of binary erasure channels, which give the following class of states
\begin{equation}
\rho = \frac{1}{2} \ketbra{}{0}{0} \otimes ((1-\epsilon)\ketbra{}{0}{0} + \epsilon\ketbra{}{e}{e}) + \frac{1}{2} \ketbra{}{1}{1} \otimes ((1-\epsilon)\ketbra{}{1}{1} + \epsilon\ketbra{}{e}{e}).
\end{equation}
\begin{remark} 
It is interesting to note -- concerning the duality relations used before -- that  the upper bound \emph{can} coincide with the quantum bound because the dual channel of a BEC with error probability $\epsilon$ is again a channel from the same family, i.e.\ a BEC with error probability $1-\epsilon$. 
\end{remark}

\subsection{Numerical evidence}
\begin{figure}[t!]
\centering
\includegraphics[clip, scale=0.9]{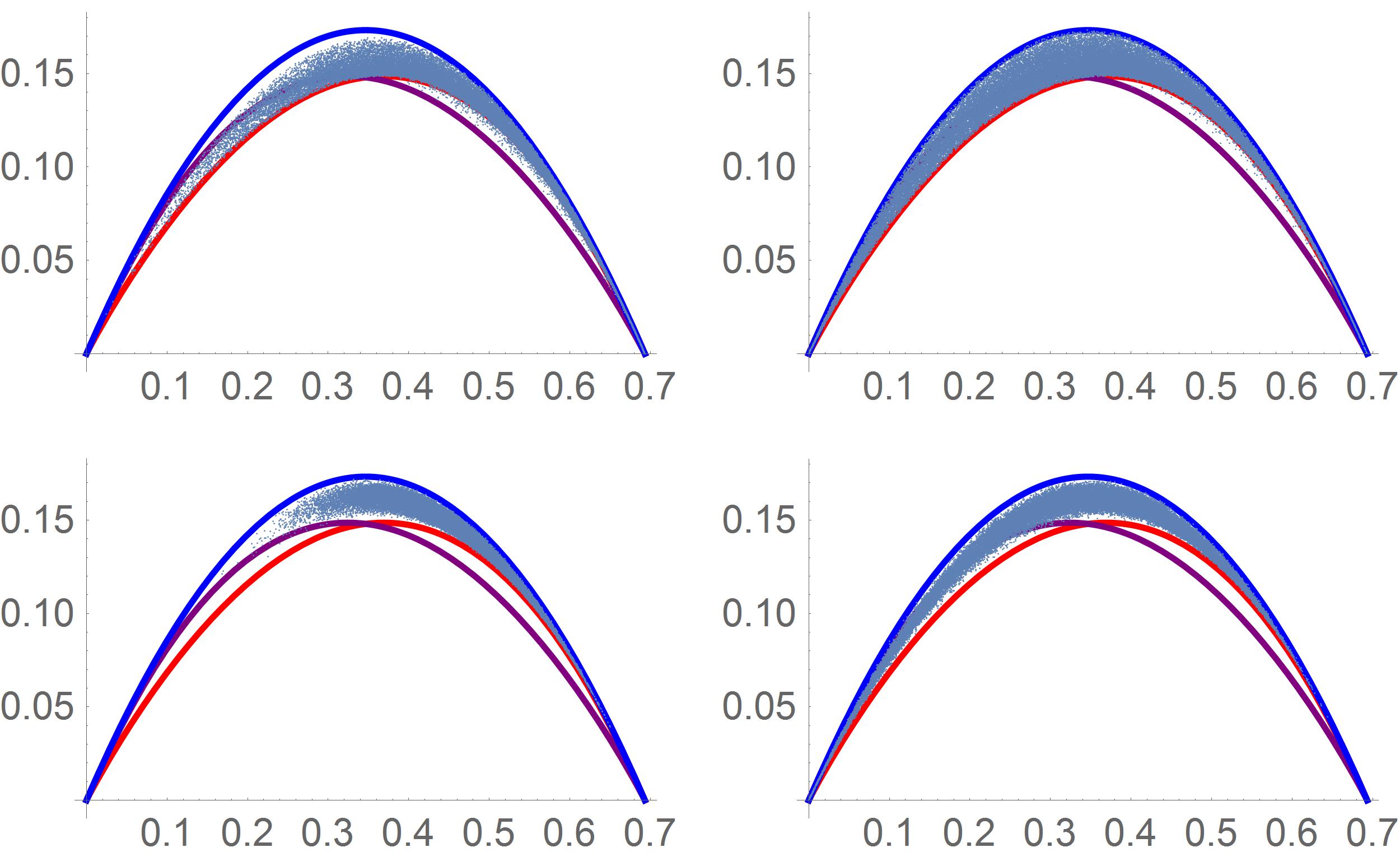}
\caption{\label{fig:numerics} These pictures present some of the numerical evidence gathered to test Conjectures~\ref{QMGL} and~\ref{upper}. Plotted is $H(X_1+X_2|B_1B_2) - H$ when $H_1=H_2=H$ against $H$. For each of the examples classical-quantum states of the form in Equation~\ref{num:test-states} where randomly generated. Those on the left with $p=\frac{1}{2}$ and those on the right with $p\in [0,1]$ at random, those on the top with quantum states of dimension $d=2$ and the bottom ones with dimension $d=4$. Each plot contains $50000$ examples.} 
\end{figure}

Naturally we tested our Conjectures~\ref{QMGL} and~\ref{upper} by numerical examples. 
For those we generated classical-quantum states of the form
\begin{equation}\label{num:test-states}
\rho^{XB}=p{\ketbra{}{0}{0}}_{X}\otimes\rho_0^{B}+(1-p){\ketbra{}{1}{1}}_{X}\otimes\rho_1^{B},
\end{equation}
where $\rho_0^{B}$ and $\rho_1^{B}$ are randomly chosen quantum states of dimension $d\in\{ 2,3,4,5,6 \}$, using the over-parametrized generation method (see e.g.~\cite{M15}),  and $p$ is either fixed $p=\frac{1}{2}$ or drawn at random from $p\in[0,1]$. We then used each of these states to calculate the exact value of $H(X_1+X_2|B_1B_2)$ with $H_1=H_2=H$ and compared it to our conjectured bounds. For all $10$ combinations we tested our conjectures with several $100000$ resulting classical-quantum states. No violations of our bounds were found. We furthermore found that the states coming close to our conjectured bounds are close to the conjectured optimal forms stated in Section \ref{subsec:optimal}.

We remark that while none of the generated states violated our conjectured bounds, violation of the classical lower bound was easily observed. A sample of our numerics is shown in Figure~\ref{fig:numerics}. (That this violation of the classical lower bound must occur ist clear from the analytic results of Section \ref{subsec:optimal}).

Additionally we carried out similar numerics for the case of two different classical-quantum states, i.e.\ with differing entropies $H_1\neq H_2$. Again, we found no violation of Conjectures~\ref{QMGL} and~\ref{upper}.

\section{Application to classical-quantum Polar codes}\label{PolarCodes}

In this section we apply the above results to classical-quantum Polar codes. We will first introduce the concept and the necessary technical aspects of Polar codes. Then we will show how our results can be used to translate a simple proof of polarization from the classical-classical case to the classical-quantum case. Our results also allow us to prove polarization for non-stationary channels. Finally, we will describe the impact of our quantitative bounds from Section \ref{lower-bound} on the \emph{speed} of polarization of cq-Polar codes and comment on the possible speed when assuming our conjectured lower bound from Conjecture \ref{QMGL}.

\subsection{Introduction to cq-Polar codes}
Polar codes were introduced by Arikan as the first classical constructive capacity achieving codes with efficient encoding and decoding~\cite{A09}. The underlying idea of Polar codes is that by adding the input bit of a later channel onto one of an earlier channel, that earlier channel becomes harder to decode while providing side-information for decoding the later one. Polar Codes rely on an iteration of this scheme, which eventually leads to almost perfect or almost useless channels, combined with a successive cancellation decoder. This decoder attempts to decode the output bit by bit, assuming at each step full knowledge of previously received bits while ignoring later outputs. Since information is sent only via channels that polarize to (almost) perfect channels while useless channels transmit so called \textit{frozen bits}, which are known to the receiver, this decoder can achieve a very low error probability. In fact, it was proven in \cite{AT09} that the block error probability scales as $O(2^{-N^\beta})$ (for any $\beta<1/2$).

Based on the classical setting, Polar codes were later generalized to channels with quantum outputs~\cite{WG13}. These quantum polar codes inherit many of the desirable features like the efficient encoder and the exponentially vanishing block error probability~\cite{WG13, H14}, while especially the efficient decoder remains an open problem~\cite{WLH13}.

Since their introduction Polar codes have been investigated in many ways, like adaptations to many different settings in classical~\cite{A12, A10} and quantum information theory~\cite{HMW14,CM15}.

In particular, in the classical setting, Polar codes have been generalized to non-stationary channels~\cite{AT14} and it was shown that the exponentially vanishing block error rate can be achieved with just a polynomial block length \cite{GX15}. Both of these results have not so far been extended to the classical-quantum setting, and their proofs rely heavily on the classical Mrs.\ Gerbers Lemma.

Let us now look at the relationship between bounds on information combining and Polar codes. The most natural quantity to track the quality of a channel during the polarization process is its conditional entropy (or equivalently, for symmetric channels, its mutual information), and the most basic element in Polar coding is the application of a CNOT gate. As described above, from such an application we can derive one channel that is worse then either of the two original channels, and one that is better (in terms of their conditional entropy). The worse channel is usually denoted by $\langle W_1,W_2\rangle^-$ and the better one by $\langle W_1,W_2\rangle^+$, where $W_1$ and $W_2$ are the original channels. It turns out that (see Section~\ref{qcombining})
\begin{equation}
H(\langle W_1,W_2\rangle^-) = H(X_1+X_2 | Y_1Y_2) = H(W_1 \boxast W_2)
\end{equation}
and
\begin{equation}
H(\langle W_1,W_2\rangle^+) = H(X_2 | X_1+X_2,Y_1Y_2) = H(W_1 \varoast W_2).
\end{equation}
Naturally, the same is true for the corresponding quantities based on the channel mutual information $I(W)$, which we recall is defined by $I(W):=\log2-H(W)$ for the case of symmetric channels, which is the only case we consider here.

Therefore, it is intuitive that good bounds on information combining can be very helpful for investigating specific properties of Polar codes and in particular the of the polarization process. This is because those bounds allow to characterize the difference in entropy between the synthesized channels $\langle W_1,W_2\rangle^-$ and $\langle W_1,W_2\rangle^+$ and the original channels $W_1,W_2$.

\subsection{Polarization for stationary and non-stationary channels}
Polarization is one of the main features of Polar codes and crucial for their ability to achieve capacity. It was first proven in the classical setting in~\cite{A09} by showing convergence of certain martingales and a similar approach has later been used to establish polarization for classical-quantum Polar codes in~\cite{WG13}. Recently a conceptually simpler proof of polarization has been found in~\cite{AT14} making use of the classical Mrs.\ Gerbers Lemma as its main tool. Besides its more intuitive approach, one of the main advantages of this new proof is that it can be extended to non-stationary channels, while the martingale approach is only known to work for stationary channels.

In this section we show that our results from Section \ref{lower-bound} are sufficient to extend the polarization proof from \cite{AT14} to the setting of classical-quantum channels, and also to prove polarization for non-stationary classical-quantum channels. The main observation that enables us to translate the classical proofs is the following Lemma.
\begin{lem}{}\label{lem:lower}
Let $W_1$ and $W_2$ be two classical-quantum binary and symmetric channels with $I(W_1),I(W_2)\in [a,b]$, then the following holds
\begin{equation}
I(\langle W_1,W_2\rangle^+) - I(\langle W_1,W_2\rangle^-) \geq |I(W_1)-I(W_2)| + \mu(a,b), \label{ent-diff}
\end{equation}
where $\mu(a,b)>0$ whenever $0<a<b<\log2$. 
\end{lem}
\begin{proof}
The statement follows from the results in Section~\ref{lower-bound}, in particular Remark~\ref{equality-condition-for-proof}. To see this, note that
\begin{align}
I(\langle W_1,W_2\rangle^+) - I(\langle W_1,W_2\rangle^-)-|I(W_1)-I(W_2)|&=2\left(H(\langle W_1,W_2\rangle^-)-\max\{H(W_1),H(W_2)\}\right)\nonumber\\
&=2\left(H(X_1+X_1|B_1B_2)-\max\{H_1,H_2\}\right),\label{herewewantauniformlowerbound}
\end{align}
where the last line is written in the notation of Remark~\ref{equality-condition-for-proof}. Since our lower bound (\ref{firstFGmainresult}) from Theorem \ref{QuantumMrsGerberTheorem2DifferentStates} is continuous in $H_1,H_2$ and equals $0$ only on the boundary, given by the condition $H_1\in\{0,\log2\}$ or $H_2\in\{0,\log2\}$, we obtain a strictly positive uniform lower bound $\mu(a,b)>0$ on Eq.\ (\ref{herewewantauniformlowerbound}) for $H_1,H_2\in[b,\log2-a]$ with $0<a<b<\log2$ (see also Fig.\ \ref{fig-bound-with-h1h2}).
\end{proof}
In the usual setting of stationary channels it is enough to consider the two original channels $W_1=W_2=W$ to be equal in which case Equation~\ref{ent-diff} simplifies to
\begin{equation}\label{defineDeltaW}
\Delta(W) := I(W^+) - I(W^-) \geq \kappa(a,b),
\end{equation}
if $I(W)\in[a,b]$. With this tool we are now ready to address the question of polarization for classical-quantum channel. First we will look at stationary channels and prove polarization in the classical-quantum setting. As mentioned before this result was already achieved in~\cite{WG13}, but we will give an alternative simple proof based on~\cite{AT14}.
\begin{thm}
For any classical-quantum BMSC $W$ and any $0<a<b<\log2,$ the following holds
\begin{align}
&\lim_{n\rightarrow\infty} \frac{1}{2^n} \#\{ s^n \in \{+,-\}^n : I(W^{s^n})\in [0,a) \} = 1-I(W)/\log2, \\
&\lim_{n\rightarrow\infty} \frac{1}{2^n} \#\{ s^n \in \{+,-\}^n : I(W^{s^n})\in [a,b] \} = 0, \\
&\lim_{n\rightarrow\infty} \frac{1}{2^n} \#\{ s^n \in \{+,-\}^n : I(W^{s^n})\in (b,\log2] \} = I(W)/\log2.
\end{align}
\end{thm}
\begin{proof}
The proof follows essentially the one in~\cite{AT14} adjusted to the classical-quantum setting considered in our work. We will nevertheless state the important steps in the proof here. 
We start with a given classical-quantum channel $W$ and arbitrary $0<a<b<\log2$. We define the following
\begin{align}
\alpha_n(a) := \frac{1}{2^n} \#\{ s \in \{+,-\}^n : I(W^{s})\in [0,a) \}, \\
\theta_n(a,b):=\frac{1}{2^n} \#\{ s \in \{+,-\}^n : I(W^{s})\in [a,b] \}, \\
\beta_n(b):=\frac{1}{2^n} \#\{ s \in \{+,-\}^n : I(W^{s})\in (b,\log2] \},
\end{align}
where $s:=s^n$ to simplify the notation. Furthermore we will need to additional quantities 
\begin{equation}
\mu_n = \frac{1}{2^n} \sum_{ s \in \{+,-\}^n} I(W^s)
\end{equation}
and
\begin{equation}
\nu_n = \frac{1}{2^n} \sum_{ s \in \{+,-\}^n} [I(W^s)]^2.
\end{equation}
Now, it follows directly from the chain rule (Equation~\ref{qchain}) that 
\begin{equation}
\mu_{n+1} = \mu_n = I(W). 
\end{equation}
It can also be seen that 
\begin{align}
\nu_{n+1} &=  \frac{1}{2^{n+1}} \sum_{ s \in \{+,-\}^{n+1}} I(W^s)^2 \\
&=  \frac{1}{2^{n}} \sum_{ t \in \{+,-\}^{n}} \frac{1}{2}[I(W^{t+})^2 + I(W^{t-})^2] \\
&=  \frac{1}{2^{n}} \sum_{ t \in \{+,-\}^{n}} I(W^{t})^2 +\left(\frac{1}{2}\Delta(W^{t})\right)^2 \\
&\geq \nu_n + \frac{1}{4}\theta_n(a,b)\kappa(a,b)^2,
\end{align}
where $\Delta(W)$ has been defined in (\ref{defineDeltaW}) and we take $\kappa(a,b)>0$ from Lemma \ref{lem:lower}. It follows that $\nu_n$ is monotonically increasing and, since it is also bounded, therefore converging. Particularly we can use it to bound $\theta_n(a,b)$ by
\begin{equation}
0\leq \theta_n(a,b) \leq 4\frac{\nu_{n+1} - \nu_n}{\kappa(a,b)^2}
\end{equation}
and therefore conclude that $\lim_{n\rightarrow\infty}\theta_n(a,b) = 0$. Next we show that 
\begin{align}
I(W) = \mu_n &\leq a \alpha_n(a) + b\theta_n(a,b) + (\log2)\beta_n(b) \\
&= a + (b-a)\theta_n(a,b) + (\log2-a)\beta_n(b),
\end{align}
thus by taking $n$ to infinity and $a$ infinitesimal small it follows that
\begin{equation}
\liminf_{n\rightarrow\infty} \beta_n(b) \geq I(W)/\log2. 
\end{equation}
Similarly upper bounding $1-\mu_n$ leads to
\begin{equation}
\liminf_{n\rightarrow\infty} \alpha_n(a) \geq 1-I(W)/\log2. 
\end{equation}
Finally the original claim follows from the fact that $\alpha_n(a) + \beta_n(b) \leq 1$.
\end{proof}
Now we will look at classical-quantum Polar codes for non-stationary channels, following the treatment in \cite{AT14}. Instead of a fixed channel $W$ we start with a collection of channels $W_{0,t}$, where the first index numbers the coding step and the second the channel position. From here we can define the coding steps similar to the classical case recursively as
\begin{align}
W_{n,Nm+j} &= \langle W_{n-1,Nm+j} ,W_{n-1,Nm+N/2+j} \rangle^- \\
W_{n,Nm+N/2+j} &= \langle W_{n-1,Nm+j} ,W_{n-1,Nm+N/2+j} \rangle^+,
\end{align}
with $n\geq 1$, $N=2^n$ and $0\leq j\leq N/2-1$.
With these definitions we can state the result for non-stationary channels.  
\begin{thm}
For any collection of classical-quantum BMSC $W_{0,t}$ and any $0<a<b<\log2$, the following holds
\begin{align}
&\lim_{n\rightarrow\infty}\lim_{T\rightarrow\infty} \frac{1}{T} \#\{ 0\leq t<T : I(W_{n,t})\in [0,a) \} = 1-\mu/\log2, \\
&\lim_{n\rightarrow\infty}\lim_{T\rightarrow\infty} \frac{1}{T} \#\{  0\leq t<T : I(W_{n,t})\in [a,b] \} = 0, \\
&\lim_{n\rightarrow\infty}\lim_{T\rightarrow\infty} \frac{1}{T} \#\{  0\leq t<T : I(W_{n,t})\in (b,\log2] \} = \mu/\log2,
\end{align}
with $\mu=\lim_{T\rightarrow\infty}\frac{1}{T} \sum_{t<T}I(W_{0,t})$, under the condition that $\mu$ is well defined.
\end{thm}
\begin{proof}
Again the proof will follow very closely the one in~\cite{AT14}. 
We start again by defining the fractions $\alpha_n(a)$, $\theta_n(a,b)$ and $\beta_n(b)$ as the quantities under investigation before taking the limes over $n$. Furthermore we will similarly to the last proof define the quantities
\begin{equation}
\mu_n=\lim_{T\rightarrow\infty}\frac{1}{T} \sum_{t<T}I(W_{n,t})
\end{equation}
and
\begin{equation}
\nu_n=\liminf_{T\rightarrow\infty}\frac{1}{T} \sum_{t<T}I(W_{n,t})^2.
\end{equation}
Note that from the assumption that the limit in $\mu = \mu_0$ exists also follows that all $\mu_n$ are well defined, the reasoning being the same as in the classical case (see~\cite{AT14}).
Therefore it also follows that $\mu_n=\mu_{n+1}$ as in the previous proof.  \\
Next we are looking at the change in variance when combining two channels. From the general Lemma~\ref{lem:lower} we can also deduce the following statement
\begin{equation}
\Delta^2(W_1,W_2) := \frac{1}{2} [I(\langle W_1,W_2\rangle^-)^2 + I(\langle W_1,W_2\rangle^+)^2] - \frac{1}{2} [I(W_1)^2 + I(W_2)^2] \geq \zeta(a,b),
\end{equation}
if $I(W_1),I(W_2)\in[a,b]$, where $\zeta(a,b)>0$ whenever $0<a<b<1$. This is sufficient to conclude that $\nu_{n+1}\geq\nu_n$, however to relate their difference to $\theta_n$ more work is needed. It is easy to see that in special cases (e.g. every second channel is already extremal) the combination of different channels might not lead to a positive $\zeta(a,b)$ bounding $\nu_{n+1}-\nu_n$. Nevertheless even those seemingly ineffective coding steps deterministically permute the channels and therefore allow for progress in later coding steps. This has been made precise in~\cite{AT14} in a Corollary that we will also use here. It states that if $\theta_n(a,b) > {{k}\choose{\left\lfloor k/2\right\rfloor}}/2^k:=\epsilon_k$, then
\begin{equation}
\nu_{n+k} \geq \nu_n + \delta,
\end{equation}
where $\delta>0$ is a quantity that depends only on $k$, $\theta_n$, $a$ and $b$. The proof in~\cite{AT14} is entirely algebraic and works also in our generalized setting. From this we can conclude, for every $k\in{\mathbb N}$, that $\theta_n\leq\epsilon_k$ holds for sufficiently large $n$, and therefore
\begin{equation}
\lim_{n\rightarrow\infty}\lim_{T\rightarrow\infty} \frac{1}{T} \#\{  0\leq t<T : I(W_{n,t})\in [a,b] \} = 0,
\end{equation}
since $\lim_{k\rightarrow\infty}\epsilon_k = 0$. \\
The claims about $\alpha_n$ and $\beta_n$ now follow from the same reasoning as in the stationary case. 
\end{proof}

\subsection{Speed of polarization}
Applying our quantitative result from Theorem \ref{QuantumMrsGerberTheorem1StateTwice} to the entropy change of binary-input classical-quantum channels under the polar transform, we now prove a \emph{quantitative} result on the speed of polarization for i.i.d.\ binary-input classical-quantum channels. For our proof, we adapt the method of \cite{GX15} to the $\sim H/(-\log H)$ lower bound guaranteed by our Eq.\ (\ref{FonlyMoreConvenientLowerBound}), which is somewhat worse than the linear lower bound $\sim H$ for the classical-classical case in \cite[see in particular Lemma 6]{GX15}; this is the reason that our following result does not guarantee a polynomial blocklength $\sim(1/\varepsilon)^\mu$, but only a subexponential one $\sim(1/\varepsilon)^{\mu\log1/\varepsilon}$. Under our Conjecture \ref{QMGL}, however, we can show the same polynomial blocklength result as in \cite{GX15} for classical-classical channels, as we will point out in Remark \ref{poly-blocklength-under-conjecture}. Note that we do not make any claim about efficient decoding of classical-quantum polar codes (e.g.\ with a circuit of subexponential size), which remains an open problem (see Section \ref{conclusion-section}).

\begin{thm}[Blocklength subexponential in gap to capacity suffices for classical-quantum binary polar codes]There is an absolute constant $\mu<\infty$ such that the following holds. For any binary-input classical-quantum channel $W$, there exists $a_W<\infty$ such that for all $\varepsilon>0$ and all powers of two $N\geq a_W(1/\varepsilon)^{\mu\log1/\varepsilon}$, a polar code of blocklength $N$ has rate at least $I(W)-\varepsilon$ and block-error probability at most $2^{-N^{0.49}}$, where $I(W)$ is the symmetric capacity of $W$.
\end{thm}
\begin{proof}Our proof follows the proofs of \cite[Propositions 5 and 10]{GX15} (``rough'' and ``fine'' polarization). The main reason why we can guarantee only a subexponential scaling here, lies in the rough polarization step (\cite[Proposition 5]{GX15}). In the following, we outline only the main differences to the proofs in \cite{GX15} which are responsible for the altered scaling. As in \cite{GX15} we define $T(W):=H(W)(1-H(W))$. Then \cite[Lemma 8]{GX15} is modified to
\begin{align*}
\underset{i~\text{mod}~2}{{\mathbb E}}[T(W^{(i)}_{n+1}]\leq T(W^{(\lfloor i/2\rfloor)}_n)-\kappa\frac{T(W^{(\lfloor i/2\rfloor)}_n)}{-\log T(W^{(\lfloor i/2\rfloor)}_n)}
\end{align*}
with some $\kappa>0$. Using convexity we obtain the same relation for the full expectation values (similar to the equation in the proof of \cite[Corollary 9]{GX15}):
\begin{align*}
\underset{i}{{\mathbb E}}[T(W^{(i)}_{n+1})]\leq \underset{i}{{\mathbb E}}[T(W^{(i)}_n)]-\kappa\frac{\underset{i}{{\mathbb E}}[T(W^{(i)}_n)]}{-\log\underset{i}{{\mathbb E}}[T(W^{(i)}_n)]}.
\end{align*}
This now does not anymore guarantee that the decrease of $\underset{i}{{\mathbb E}}[T(W^{(i)}_n)]$ is exponential in $n$, as in \cite[Corollary 9]{GX15} which was obtained from the recursion $\underset{i}{{\mathbb E}}[T(W^{(i)}_{n+1})]\leq \underset{i}{{\mathbb E}}[T(W^{(i)}_n)]-\kappa\underset{i}{{\mathbb E}}[T(W^{(i)}_n)]$ (or the same recursion for $\underset{i}{{\mathbb E}}[\sqrt{T(W^{(i)}_{n})}]$). Thus, instead of the differential equation $\frac{d}{dn}f(n)=-\kappa f(n)$, the behaviour here is goverened by the equation $\frac{d}{dn}f(n)=-\kappa\frac{f(n)}{-\log f(n)}$. This differential equation has the solution $f(n)=\exp[-\sqrt{2\kappa n+(\log f(0))^2}]$ (note, $f(n)\leq1$ for all $n$) und we therefore obtain the following bound:
\begin{align*}
\underset{i}{{\mathbb E}}[T(W^{(i)}_n)]\leq e^{-\sqrt{2\kappa n+(\log T(W_0^{(0)}))^2}}\leq e^{-\sqrt{2\kappa n}},
\end{align*}
guaranteeing the expectation value of $T(W^{(i)}_n)$ to decrease at least superpolynomially with the number of polarization steps $n$.

This expectation value will thus be smaller than any $\delta>0$ if only the number of polarization steps satisfies $n\geq\frac{1}{2\kappa}\left(\log\frac{1}{\delta}\right)^2\sim\left(\log\frac{1}{\delta}\right)^2$. This expression can now be connected with the ``fine polarization step'' \cite[Proposition 10]{GX15} since for any fixed power $\delta\sim\varepsilon^p$ (with $\varepsilon$ from the statement of the theorem) we again obtain that $n\geq\widetilde{\mu}\left(\log\frac{1}{\varepsilon}\right)^2$ with some constant $\widetilde{\mu}$ suffices. Since the number $n$ of polarization steps is related to the blocklength $N$ via $N=2^n$, we find that the constructed polar code has the desired properties as soon as the blocklength satisfies $N\geq2^{\widetilde{\mu}(\log1/\varepsilon)^2}=(1/\varepsilon)^{\mu\log1/\varepsilon}$ (with $\mu=\widetilde{\mu}\log2$). The constant $a_W$ from the theorem statement accounts for the fact that the above analysis is only valid for sufficiently small $\varepsilon$.

It is instructive to compare the reasoning in the previous paragraph with the blocklength result obtained in \cite{GX15}. The bound obtained from $f(n)$ in this case is $\underset{i}{{\mathbb E}}[T(W^{(i)}_n)]\leq e^{-\kappa n}$, so that $n\geq\widetilde{\mu}\log\frac{1}{\varepsilon}$ suffices for $\underset{i}{{\mathbb E}}[T(W^{(i)}_n)]\leq\varepsilon^p$. This shows that a blocklength $N\geq2^{\widetilde{\mu}\log1/\varepsilon}=(1/\varepsilon)^{\mu}$ suffices.
\end{proof}

\begin{remark}[Polynomial blocklength suffices under Conjecture \ref{QMGL}]\label{poly-blocklength-under-conjecture}
If Conjecture \ref{QMGL} holds, then one can prove the same polynomial blocklength result as \cite[Theorem 1]{GX15} also for classical-quantum channels. The only part of the proof which has to be changed is \cite[Lemma 6]{GX15}, where the classical Mrs.\ Gerber's Lemma is to be replaced by Conjecture \ref{QMGL}. However, this change does not even affect the numerical value of $\theta$ that can be chosen in \cite[Lemma 6]{GX15}, since our conjectured optimal lower bound in the classical-quantum case is simply a symmetrization of the classical lower bound.
\end{remark}

\section{Conclusion and open questions}\label{conclusion-section}
In this work we have investigated the problem of bounds on information combining when the side information available is quantum. This is a generalization of the classical problem of information combining which has found many applications in classical information theory. 
In particular we find a non-trivial lower bound the conditional entropy of check nodes (or the minus polar transform) and accordingly an upper bound on that of variable nodes (or the plus polar transform). 

On the way of proofing this bound we find several technical results that are also of interest in their own merit, including a novel lower bound on the concavity of von Neumann entropy which we expect to also be useful in many other scenarios whenever a bound in terms of the fidelity is needed. 
Furthermore we show a direct relation between our problem and lower bounds on the conditional mutual information which have generated much attention recently. Our proof gives a direct application of the most prominent result, the lower bound by Fawzi and Renner in terms of the Fidelity of Recovery~\cite{FR14}. Nevertheless it also raises the question whether there are stronger bounds especially in the case when the conditioned systems are classical, which would allow to get closer to our conjectured bounds.
Another important ingredient is the concept of channel duality. The fact that duality is not only useful in proving our bound, but also provides an intuitive explanation for the states that achieve our conjecture with equality (and therefore the additional symmetry in our conjecture), might point to a close relation between these two fields. 

Finally, the application of our bounds to classical-quantum Polar coding allows us to prove new results, namely that non-stationary classical-quantum channels also experience polarization and that a sub-exponential block length is sufficient  to achieve the optimal block-error rate for stationary Polar codes. 

In the same manner we expect our results to also have applications in other coding scenarios, such as branching MERA and convolutional Polar codes~\cite{FP14, FP142, FHP17}. In general the applications of the classical bounds give natural possibilities for quantum extensions. 

Lastly we would like to point out some open problems. The most obvious one being to find a proof for our conjectured bounds (see Section \ref{main}), which comes along with several other open questions, such as a better understanding of conditioning on quantum system and duality in quantum information theory as well as new bounds on strong subadditivity. 
Also our given lower bound as well as the conjectured ones can be seen as special cases of the Mrs. Gerbers Lemma by Wyner and Ziv, which in their version not only applies to single copies of the channel but $n$ copies. Since its discovery the Mrs. Gerbers Lemma has been generalized to many settings~\cite{W74, AK77, JA12, OS15, C14}, all of which pose natural open problems in the quantum setting. While the $n$-copy case could be useful in Shannon theory, generalization to non-binary inputs would have applications to coding such as Polar codes for arbitrary classical-quantum channels (see e.g~\cite{GV14,GB15,NR17}). A final question is whether an efficient decoding procedure (e.g., using a number of gates polynomial in the blocklength) exists for classical-quantum Polar codes~\cite{WLH13}.

\section*{Acknowledgments}
We thank Robert K\"onig, Christian Majenz, and Joe Renes for fruitful discussions, and Joe in particular for pointing us to the relation between our problem and channel duality. \\
CH acknowledges support by the Spanish MINECO, project FIS2013-40627-P, FIS2016-80681-P (AEI/FEDER, UE) and FPI Grant No. BES-2014-068888, as well as by the Generalitat de Catalunya, CIRIT project no. 2014 SGR 966. \\
DR acknowledges support from the ERC grant DQSIM.

\bibliographystyle{IEEEtran}
\bibliography{bib}

% Generated by IEEEtran.bst, version: 1.14 (2015/08/26)
\begin{thebibliography}{10}
\providecommand{\url}[1]{#1}
\csname url@samestyle\endcsname
\providecommand{\newblock}{\relax}
\providecommand{\bibinfo}[2]{#2}
\providecommand{\BIBentrySTDinterwordspacing}{\spaceskip=0pt\relax}
\providecommand{\BIBentryALTinterwordstretchfactor}{4}
\providecommand{\BIBentryALTinterwordspacing}{\spaceskip=\fontdimen2\font plus
\BIBentryALTinterwordstretchfactor\fontdimen3\font minus
  \fontdimen4\font\relax}
\providecommand{\BIBforeignlanguage}[2]{{%
\expandafter\ifx\csname l@#1\endcsname\relax
\typeout{** WARNING: IEEEtran.bst: No hyphenation pattern has been}%
\typeout{** loaded for the language `#1'. Using the pattern for}%
\typeout{** the default language instead.}%
\else
\language=\csname l@#1\endcsname
\fi
#2}}
\providecommand{\BIBdecl}{\relax}
\BIBdecl

\bibitem{WZ73}
A.~D. Wyner and J.~Ziv, ``A theorem on the entropy of certain binary sequences
  and applications--i,'' \emph{IEEE Transactions on Information Theory},
  vol.~19, no.~6, pp. 769 -- 772, November 1973.

\bibitem{GKbook}
A.~E. Gamal and Y.-H. Kim, \emph{Network Information Theory}.\hskip 1em plus
  0.5em minus 0.4em\relax New York, U. S. A.: Cambridge University Press,
  January 2012.

\bibitem{RU08}
T.~Richardson and R.~Urbanke, \emph{Modern Coding Theory}.\hskip 1em plus 0.5em
  minus 0.4em\relax New York, NY, USA: Cambridge University Press, 2008.

\bibitem{A09}
E.~Arikan, ``Channel polarization: A method for constructing capacity-achieving
  codes for symmetric binary-input memoryless channels,'' \emph{IEEE
  Transactions on Information Theory}, vol.~55, no.~7, pp. 3051--3073, July
  2009, arXiv:0807.3917.

\bibitem{AT14}
M.~Alsan and E.~Teletar, ``A simple proof of polarization and polarization for
  non-stationary channels,'' \emph{Proceedings of the 2014 IEEE International
  Symposium on Information Theory}, pp. 301 -- 305, July 2014.

\bibitem{GX15}
V.~Guruswami and P.~Xia, ``Polar codes: Speed of polarization and polynomial
  gap to capacity,'' \emph{IEEE Transactions on Information Theory}, vol.~61,
  no.~1, pp. 3--16, January 2015, arXiv:1304.4321.

\bibitem{WG13}
M.~M. Wilde and S.~Guha, ``Polar codes for classical-quantum channels,''
  \emph{IEEE Transactions on Information Theory}, vol.~59, no.~2, pp.
  1175--1187, February 2013, arXiv:1109.2591.

\bibitem{FR14}
O.~Fawzi and R.~Renner, ``{Quantum conditional mutual information and
  approximate Markov chains},'' Oct. 2014, arXiv:1410.0664.

\bibitem{RSH14}
J.~M. {Renes}, D.~{Sutter}, and S.~{Hamed Hassani}, ``{Alignment of Polarized
  Sets},'' \emph{ArXiv e-prints}, Nov. 2014.

\bibitem{R17}
J.~M. Renes, ``{Duality of channels and codes},'' Jan. 2017, arXiv:1701.05583.

\bibitem{S48}
C.~E. Shannon, ``{A mathematical theory of communication},'' \emph{Bell System
  Tech. J.}, vol.~27, pp. 379–423, 623–656, 1948.

\bibitem{S59}
A.~J. Stam, ``Some inequalities satisfied by the quantities of information of
  fisher and shannon,'' \emph{Information and Control.}, vol.~2, no.~2, pp.
  101--112, 1959.

\bibitem{SW90}
S.~Shamai and A.~D. Wyner, ``A binary analog to the entropy-power inequality,''
  \emph{IEEE Transactions on Information Theory}, vol.~36, no.~6, pp.
  1428--1430, November 1990.

\bibitem{HAT14}
S.~Haghighatshoar, E.~Abbe, and {\`I}.~E. Telatar, ``A new entropy power
  inequality for integer-valued random variables,'' \emph{IEEE Transactions on
  Information Theory}, vol.~60, no.~7, pp. 3787--3796, 2014.

\bibitem{Tao10}
T.~Tao, ``Sumset and inverse sumset theory for shannon entropy,''
  \emph{Combinatorics, Probability and Computing}, vol.~19, no.~04, pp.
  603--639, 2010.

\bibitem{KS14}
R.~{Koenig} and G.~{Smith}, ``{The entropy power inequality for quantum
  systems},'' \emph{IEEE Transactions on Information Theory}, vol.~60, no.~3,
  pp. 1536--1548, March 2014.

\bibitem{PMG14}
G.~{de Palma}, A.~{Mari}, and V.~{Giovannetti}, ``{A generalization of the
  entropy power inequality to bosonic quantum systems},'' \emph{Nature
  Photonics}, vol.~8, pp. 958--964, Dec. 2014.

\bibitem{ADO15}
K.~{Audenaert}, N.~{Datta}, and M.~{Ozols}, ``{Entropy power inequalities for
  qudits},'' \emph{Journal of Mathematical Physics}, vol.~57, no.~5, p. 052202,
  May 2016.

\bibitem{NC00}
M.~A. Nielsen and I.~L. Chuang, \emph{Quantum Information and Quantum
  Computation}.\hskip 1em plus 0.5em minus 0.4em\relax Cambridge University
  Press, 2000.

\bibitem{K15}
R.~{Koenig}, ``{The conditional entropy power inequality for Gaussian quantum
  states},'' \emph{Journal of Mathematical Physics}, vol.~56, no.~2, p. 022201,
  Feb. 2015.

\bibitem{dPT17}
G.~{De Palma} and D.~{Trevisan}, ``{The Entropy Power Inequality with quantum
  memory},'' \emph{ArXiv e-prints}, 2017.

\bibitem{LHHH05}
I.~Land, S.~Huettinger, P.~A. Hoeher, and J.~B. Huber, ``{Bounds on Information
  Combining},'' \emph{IEEE Transactions on Information Theory}, vol.~51, no.~2,
  pp. 612--619, Feb. 2005.

\bibitem{SSZ05}
I.~Sutskover, S.~Shamai, and J.~Ziv, ``{Extremes of information combining},''
  \emph{IEEE Transactions on Information Theory}, vol.~51, no.~4, pp. 1313 --
  1325, 2005.

\bibitem{R16bp}
J.~M. {Renes}, ``{Belief propagation decoding of quantum channels by passing
  quantum messages},'' \emph{ArXiv e-prints}, Jul. 2016.

\bibitem{Lieb2002}
\BIBentryALTinterwordspacing
E.~H. Lieb and M.~B. Ruskai, \emph{Proof of the strong subadditivity of
  quantum-mechanical entropy}.\hskip 1em plus 0.5em minus 0.4em\relax Berlin,
  Heidelberg: Springer Berlin Heidelberg, 2002, pp. 63--66. [Online].
  Available: \url{http://dx.doi.org/10.1007/978-3-642-55925-9_6}
\BIBentrySTDinterwordspacing

\bibitem{DFR16}
F.~{Dupuis}, O.~{Fawzi}, and R.~{Renner}, ``{Entropy accumulation},''
  \emph{ArXiv e-prints}, Jul. 2016.

\bibitem{WR12}
J.~M. Renes and M.~M. Wilde, ``Polar codes for private and quantum
  communication over arbitrary channels,'' \emph{IEEE Transactions on
  Information Theory}, vol.~60, no.~6, pp. 3090--3103, June 2014,
  arXiv:1212.2537.

\bibitem{RB08}
\BIBentryALTinterwordspacing
J.~M. Renes and J.-C. Boileau, ``Physical underpinnings of privacy,''
  \emph{Physical Review A}, vol.~78, p. 032335, Sep 2008. [Online]. Available:
  \url{http://link.aps.org/doi/10.1103/PhysRevA.78.032335}
\BIBentrySTDinterwordspacing

\bibitem{WR12a}
M.~M. Wilde and J.~M. Renes, ``Quantum polar codes for arbitrary channels,''
  \emph{Proceedings of the 2012 IEEE International Symposium on Information
  Theory}, pp. 334--338, July 2012.

\bibitem{Wilde2014asdf}
\BIBentryALTinterwordspacing
M.~M. Wilde, A.~Winter, and D.~Yang, ``Strong converse for the classical
  capacity of entanglement-breaking and hadamard channels via a sandwiched
  r{\'e}nyi relative entropy,'' \emph{Communications in Mathematical Physics},
  vol. 331, no.~2, pp. 593--622, 2014. [Online]. Available:
  \url{http://dx.doi.org/10.1007/s00220-014-2122-x}
\BIBentrySTDinterwordspacing

\bibitem{muller2013quantum}
M.~M{\"u}ller-Lennert, F.~Dupuis, O.~Szehr, S.~Fehr, and M.~Tomamichel, ``On
  quantum r{\'e}nyi entropies: a new generalization and some properties,''
  \emph{Journal of Mathematical Physics}, vol.~54, no.~12, p. 122203, 2013.

\bibitem{tomamichel2015quantum}
M.~Tomamichel, \emph{Quantum Information Processing with Finite Resources:
  Mathematical Foundations}.\hskip 1em plus 0.5em minus 0.4em\relax Springer,
  2015, vol.~5.

\bibitem{audenaert2012comparisons}
K.~M. Audenaert, ``Comparisons between quantum state distinguishability
  measures,'' \emph{Quant. Inf. Comp.}, vol.~14, no.~1, pp. 31--38, 2014.

\bibitem{W13}
M.~M. Wilde, \emph{Quantum Information Theory}.\hskip 1em plus 0.5em minus
  0.4em\relax Cambridge University Press, 2013.

\bibitem{PhysRevLett.105.040505}
\BIBentryALTinterwordspacing
W.~Roga, M.~Fannes, and K.~\ifmmode~\dot{Z}\else \.{Z}\fi{}yczkowski,
  ``Universal bounds for the holevo quantity, coherent information, and the
  jensen-shannon divergence,'' \emph{Phys. Rev. Lett.}, vol. 105, p. 040505,
  Jul 2010. [Online]. Available:
  \url{http://link.aps.org/doi/10.1103/PhysRevLett.105.040505}
\BIBentrySTDinterwordspacing

\bibitem{kim2014bounds}
I.~Kim and M.~B. Ruskai, ``Bounds on the concavity of quantum entropy,''
  \emph{Journal of Mathematical Physics}, vol.~55, no.~9, p. 092201, 2014.

\bibitem{MFW16}
A.~M{\"u}ller-Hermes, D.~Stilck~Fran{\c{c}}a, and M.~M. Wolf, ``Relative
  entropy convergence for depolarizing channels,'' \emph{Journal of
  Mathematical Physics}, vol.~57, no.~2, p. 022202, 2016.

\bibitem{NR17}
R.~Nasser and J.~M. Renes, ``{Polar Codes for Arbitrary Classical-Quantum
  Channels and Arbitrary cq-MACs},'' Jan. 2017, arXiv:1701.03397.

\bibitem{SRDR13}
D.~Sutter, J.~M. Renes, F.~Dupuis, and R.~Renner, ``Efficient quantum polar
  codes requiring no preshared entanglement,'' \emph{IEEE Transactions on
  Information Theory}, vol.~61, no.~11, pp. 6395 -- 6414, November 2015,
  arXiv:1307.1136.

\bibitem{JRSWW15}
M.~Junge, R.~Renner, D.~Sutter, M.~M. Wilde, and A.~Winter, ``Universal
  recovery from a decrease of quantum relative entropy,'' \emph{arXiv preprint
  arXiv:1509.07127}, 2015.

\bibitem{STH16}
D.~{Sutter}, M.~{Tomamichel}, and A.~W. {Harrow}, ``{Strengthened Monotonicity
  of Relative Entropy via Pinched Petz Recovery Map},'' \emph{IEEE Transactions
  on Information Theory}, vol.~62, no.~5, pp. 2907 -- 2913, 2016.

\bibitem{SBT16}
D.~Sutter, M.~Berta, and M.~Tomamichel, ``Multivariate trace inequalities,''
  \emph{Communications in Mathematical Physics}, pp. 1--22, 2016.

\bibitem{BHOS15}
F.~G.~S.~L. {Brand{\~a}o}, A.~W. {Harrow}, J.~{Oppenheim}, and S.~{Strelchuk},
  ``{Quantum Conditional Mutual Information, Reconstructed States, and State
  Redistribution},'' \emph{Physical Review Letters}, vol. 115, no.~5, p.
  050501, Jul. 2015.

\bibitem{CHMOSWW16}
T.~{Cooney}, C.~{Hirche}, C.~{Morgan}, J.~P. {Olson}, K.~P. {Seshadreesan},
  J.~{Watrous}, and M.~M. {Wilde}, ``{Operational meaning of quantum measures
  of recovery},'' \emph{Physical Review A}, vol.~94, no.~2, p. 022310, Aug.
  2016.

\bibitem{B69}
D.~Bures, ``{An extension of Kakutani's theorem on infinite product measures to
  the tensor product of semifinite $ w\sp{\ast} $-algebras},'' \emph{Trans.
  Amer. Math. Soc.}, vol. 135, pp. 199--212, 1969.

\bibitem{GLN05}
A.~Gilchrist, N.~K. Langford, and M.~A. Nielsen, ``Distance measures to compare
  real and ideal quantum processes,'' \emph{Physical Review A}, vol.~71, no.~6,
  p. 062310, 2005.

\bibitem{R02}
A.~Rastegin, ``Relative error of state-dependent cloning,'' \emph{Physical
  Review A}, vol.~66, no.~4, p. 042304, 2002.

\bibitem{BFT15}
M.~{Berta}, O.~{Fawzi}, and M.~{Tomamichel}, ``{On Variational Expressions for
  Quantum Relative Entropies},'' \emph{ArXiv e-prints}, Dec. 2015.

\bibitem{M15}
J.~{Maziero}, ``{Random Sampling of Quantum States: a Survey of Methods. And
  Some Issues Regarding the Overparametrized Method},'' \emph{Brazilian Journal
  of Physics}, vol.~45, pp. 575--583, Dec. 2015.

\bibitem{AT09}
E.~Arikan and E.~Telatar, ``On the rate of channel polarization,''
  \emph{Proceedings of the 2009 IEEE International Symposium on Information
  Theory}, pp. 1493 -- 1495, July 2009.

\bibitem{H14}
C.~Hirche, ``{Polar codes in quantum information theory},'' 2014, {Master's
  thesis, Hannover, arXiv:1501.03737}.

\bibitem{WLH13}
M.~M. Wilde, O.~Landon-Cardinal, and P.~Hayden, ``Towards efficient decoding of
  classical-quantum polar codes,'' \emph{Proceedings of the 8th Conference on
  the Theory of Quantum Computation, Communication and Cryptography (TQC
  2013)}, pp. 157--177, May 2013, arXiv:1302.0398.

\bibitem{A12}
E.~Arikan, ``Polar coding for the {Slepian-Wolf} problem based on monotone
  chain rules,'' \emph{Proceedings of the 2012 IEEE International Symposium on
  Information Theory}, pp. 566--570, July 2012.

\bibitem{A10}
------, ``{Source polarization},'' Jan. 2010, arXiv:1001.3087.

\bibitem{HMW14}
C.~Hirche, C.~Morgan, and M.~M. Wilde, ``Polar codes in network quantum
  information theory,'' \emph{IEEE Transactions on Information Theory},
  vol.~62, no.~2, pp. 1--10, February 2016, arXiv:1409.7246.

\bibitem{CM15}
C.~Hirche and C.~Morgan, ``An improved rate region for the classical-quantum
  broadcast channel,'' \emph{Proceedings of the 2015 IEEE International
  Symposium on Information Theory}, pp. 2782 -- 2786, July 2015.

\bibitem{FP14}
A.~Ferris and D.~Poulin, ``Branching {MERA} codes: A natural extension of
  classical and quantum polar codes,'' \emph{Proceedings of the 2014 IEEE
  International Symposium on Information Theory}, pp. 1081 -- 1085, July 2014.

\bibitem{FP142}
A.~J. Ferris and D.~Poulin, ``{Tensor networks and quantum error correction},''
  \emph{Physical Review Letters}, vol. 113, no.~3, p. 030501, Jul. 2014,
  arXiv:1312.4578.

\bibitem{FHP17}
A.~J. Ferris, C.~Hirche, and D.~Poulin, ``{Convolutional Polar Codes},''
  \emph{ArXiv e-prints}, Apr. 2017.

\bibitem{W74}
H.~S. {Witsenhausen}, ``{Entropy inequalities for discrete channels},''
  \emph{IEEE Transactions on Information Theory}, vol.~20, no.~5, pp. 610--616,
  September 2074.

\bibitem{AK77}
R.~Ahlswede and J.~Körner, ``{On the connection between the entropies of input
  and output distributions of discrete memoryless channels},'' pp. 13--22,
  1977.

\bibitem{JA12}
V.~{Jog} and V.~{Anantharam}, ``{The Entropy Power Inequality and Mrs. Gerber's
  Lemma for Abelian Groups of Order 2\^{}n},'' \emph{ArXiv e-prints}, Jul.
  2012.

\bibitem{OS15}
O.~{Ordentlich} and O.~{Shayevitz}, ``{Minimum MS. E. Gerber's Lemma},''
  \emph{ArXiv e-prints}, May 2015.

\bibitem{C14}
F.~{Cheng}, ``{Generalization of Mrs. Gerber's Lemma},'' \emph{ArXiv e-prints},
  Sep. 2014.

\bibitem{GV14}
V.~Guruswami and A.~Velingker, ``{An entropy sumset inequality and polynomially
  fast convergence to Shannon capacity over all alphabets},'' Nov. 2014,
  arXiv:1411.6993.

\bibitem{GB15}
D.~Goldin and D.~Burshtein, ``On the finite length scaling of ternary polar
  codes,'' \emph{Proceedings of the 2015 IEEE International Symposium on
  Information Theory}, pp. 226--230, July 2015.

\end{thebibliography}

\end{document}